\documentclass[11.5pt,reqno]{amsart}
\usepackage[margin=1in]{geometry}
\usepackage[T1]{fontenc}
\usepackage{lmodern}
\usepackage{microtype}
\usepackage{amsmath,amssymb,amsthm,mathtools}
\usepackage{bm}
\usepackage{hyperref}
\usepackage{comment}
\usepackage{enumitem}
\usepackage{booktabs}
\hypersetup{colorlinks=true,linkcolor=blue,citecolor=blue,urlcolor=blue}

\title{Optimal Condition Numbers for Low-Rank Positive Semidefinite Matrix Sensing}
\author{Mingxuan Sun}
\address{State Key Laboratory of Mathematical Sciences, Academy of Mathematics and Systems Science, Chinese Academy of Sciences, Beijing, 100190, China; School of Mathematical Sciences, University of Chinese Academy of Sciences, Beijing, 100091, China}
\email{sunmingxuan@amss.ac.cn}

\author{Zhiqiang Xu}
\address{State Key Laboratory of Mathematical Sciences, Academy of Mathematics and Systems Science, Chinese Academy of Sciences, Beijing, 100190, China; School of Mathematical Sciences, University of Chinese Academy of Sciences, Beijing, 100091, China}
\email{xuzq@lsec.cc.ac.cn}

\newcommand{\Sn}{\mathbb{S}_n}
\newcommand{\Snp}{\mathbb{S}_n^+}
\newcommand{\Mr}{\mathcal{M}_r}
\newcommand{\A}{\mathcal{A}}
\newcommand{\Aone}{\mathcal{A}^{(1)}}
\newcommand{\PhiA}{\Phi_{\A}}
\newcommand{\B}{\mathcal{B}}
\newcommand{\PhiB}{\Phi_{\B}}
\newcommand{\ip}[2]{\left\langle #1,#2\right\rangle}
\newtheorem{theorem}{Theorem}[section]
\newtheorem{lemma}{Lemma}[section]
\newtheorem{remark}{Remark}[section]
\newtheorem{proposition}{Proposition}[section]
\newtheorem{definition}{Definition}[section]

\numberwithin{equation}{section}
\DeclareMathOperator{\sgn}{sgn}
\DeclareMathOperator{\E}{\mathbb{E}}
\DeclareMathOperator{\diag}{diag}
\newcommand{\R}{\mathbb{R}}
\newcommand{\Pbb}{\mathbb{P}}
\newcommand{\HH}{\mathbb{H}}
\newcommand{\C}{\mathbb{C}}

\newcommand{\rank}{\operatorname{rank}}
\newcommand{\norm}[1]{\left\lVert#1\right\rVert}
\newcommand{\abs}[1]{\left|#1\right|}
\newcommand{\eps}{\varepsilon}
\DeclareMathOperator{\Tr}{Tr}

\begin{document}
\begin{abstract}
In this paper we focus on the stability of positive semidefinite matrix sensing maps $\PhiA(X)=(\ip{A_i}{X})_{i=1}^m$ where $A_i\succeq 0$, $X\succeq0$ and $\rank(X)\le r$. We introduce the bi-Lipschitz constants of $\PhiA$ and define the global condition numbers as the ratio of upper and lower Lipschitz constants. We give deterministic universal lower bounds for these condition numbers that depend only on the rank $r$ and on the underlying field.
We then show that rank-one Gaussian measurements asymptotically attain these bounds, except for the operator-norm condition number over $\mathbb R$ when $r=1$. For this exceptional case, we establish a sharper universal lower bound and prove its asymptotic sharpness. As an application, we derive stability guarantees for an $\ell_1$-residual PhaseLift-type estimator at the optimal sampling scale. These results identify the intrinsic limits of conditioning in positive semidefinite matrix sensing and connect these limits to stable low-rank recovery.
\end{abstract}

\maketitle

\section{Introduction}
\subsection{Problem setup}
Matrix sensing aims to recover a structured matrix from a collection of linear measurements.
In many applications, the target matrix is assumed to be low-rank and positive semidefinite, and the sensing matrices are also positive semidefinite.
Let $\HH$ be either $\R$ or $\C$.
Set
\[
\begin{gathered}
\Sn(\HH):=\{X\in\HH^{n\times n}:X=X^*\},\,
\Snp(\HH):=\{X\in\Sn(\HH):X\succeq 0\},\,\text{ and }
\Mr(\HH):=\{X\in\Snp(\HH):\rank(X)\le r\},
\end{gathered}
\]
where $X\succeq 0$ means that $X$ is positive semidefinite.
For a family of positive semidefinite sensing matrices $\A=\{A_i\}_{i=1}^m\subset \Snp(\HH)$,
we define the measurement map $\PhiA: \Sn(\HH)\rightarrow {\R}^m$ by
\begin{equation}\label{eq:sensingmap}
\PhiA(X):= (\ip{A_1}{X},\dots,\ip{A_m}{X} ),
\end{equation}
where $\ip{A_i}{X}:=\Tr(A_i^*X)$.
Since $A_i\succeq 0$, each coordinate $\ip{A_i}{X}$ is nonnegative for $X\succeq 0$. Therefore, \eqref{eq:sensingmap} provides a natural model for positive semidefinite matrix sensing, which appears in several applications.

The model \eqref{eq:sensingmap} includes intensity-based phase retrieval
\cite{CandesLi2014} as a special case.
Indeed, when the sensing matrices satisfy $A_i=\bm{a}_i\bm{a}_i^*$
and the unknown matrix has the form $X=\bm{x}\bm{x}^*$,
with $\bm{a}_i,\bm{x}\in\HH^n$, we have
\begin{equation}
\ip{A_i}{X}
=\Tr(\bm{a}_i\bm{a}_i^*\bm{x}\bm{x}^*)
=\abs{\bm{a}_i^*\bm{x}}^2.
\label{eq:phasemap}
\end{equation}
Measurements of the form \eqref{eq:sensingmap} also occur in a number of applications besides phase retrieval, such as covariance sketching \cite{wang2025lowrankcovariance}, quantum state tomography \cite{qin2024quantum}, and compressive power spectrum estimation \cite{AD12}.

The stability of the sensing map is of vital importance in matrix sensing.
It guarantees robustness of the recovery process and provides quantitative control of the sensing map.
In order to quantify the stability of the matrix sensing map for a given set of measurement matrices,
we introduce the lower and upper Lipschitz constants of the sensing map $\PhiA$.
We use the $\ell_1$ norm to measure the distance between sensing outputs and define the lower and upper Lipschitz constants with respect to the nuclear norm and the operator norm, respectively.
For the nuclear norm, we define
\begin{equation}
L_{\A}^{\HH,*}(r):=\inf_{\substack{X,Y\in\Mr\\ X\neq Y}}\frac{\norm{\PhiA(X)-\PhiA(Y)}_1}{\norm{X-Y}_*},
\qquad
U_{\A}^{\HH,*}(r):=\sup_{\substack{X,Y\in\Mr\\ X\neq Y}}\frac{\norm{\PhiA(X)-\PhiA(Y)}_1}{\norm{X-Y}_*},
\label{eq:nuclearbounds}
\end{equation}
where $\norm{\cdot}_*$ stands for the nuclear norm. For the operator norm, we define
\begin{equation}
L_{\A}^{\HH,2}(r):=\inf_{\substack{X,Y\in\Mr\\ X\neq Y}}\frac{\norm{\PhiA(X)-\PhiA(Y)}_1}{\norm{X-Y}_2},
\qquad
U_{\A}^{\HH,2}(r):=\sup_{\substack{X,Y\in\Mr\\ X\neq Y}}\frac{\norm{\PhiA(X)-\PhiA(Y)}_1}{\norm{X-Y}_2},
\label{eq:operatorbounds}
\end{equation}
where $\norm{\cdot}_2$ stands for the operator norm. In this paper, we focus on the following condition numbers associated with the map $\PhiA$:
\begin{equation}\label{eq:conditionratio}
\beta_{\A}^{\HH,*}(r):=\frac{U_{\A}^{\HH,*}(r)}{L_{\A}^{\HH,*}(r)}
 \text{\quad and \quad}
\beta_{\A}^{\HH,2}(r):=\frac{U_{\A}^{\HH,2}(r)}{L_{\A}^{\HH,2}(r)}.
\end{equation}
If the  map $\PhiA$ fails to be injective on $\Mr(\HH)$, equivalently, $L_{\A}^{\HH,*}(r)=L_{\A}^{\HH,2}(r)=0$, then we set $\beta_{\A}^{\HH,*}(r)=\beta_{\A}^{\HH,2}(r)=+\infty$.

These quantities can be viewed as analogues of the extremal singular value ratios of a matrix:
$L_{\A}^{\HH,*}(r)$ and $L_{\A}^{\HH,2}(r)$ characterize the smallest global separation induced by the sensing map,
whereas $U_{\A}^{\HH,*}(r)$ and $U_{\A}^{\HH,2}(r)$ characterize the largest ones.
Consequently, $\beta_{\A}^{\HH,*}(r)$ and $\beta_{\A}^{\HH,2}(r)$ provide scale-invariant measures of the stability of the sensing map.
Related Lipschitz stability constants have been extensively investigated in the phase retrieval literature
\cite{BandeiraCahillMixonNelson2014,BalanWang2015,BalanZou2016,AlharbiAlshabhiFreemanGhoreishi2024,XiaXuXu2024,PengHanHuang2025}.
However, to our knowledge, analogous questions for general positive semidefinite matrix sensing remain largely unexplored.

Naturally, a fundamental question arises concerning the optimal
conditioning of positive semidefinite matrix sensing.

\noindent\textbf{Question:}
{\itshape
What are the nontrivial absolute lower bounds
\[
\beta_{0}^{\HH,*}(r)
\quad\text{and}\quad
\beta_{0}^{\HH,2}(r)
\]
for which the inequalities
\[
\beta_{\A}^{\HH,*}(r)\geq \beta_{0}^{\HH,*}(r),
\qquad
\beta_{\A}^{\HH,2}(r)\geq \beta_{0}^{\HH,2}(r)
\]
hold uniformly over all sensing families $\A$, regardless of
the number of measurements?
In other words, what is the intrinsic limitation on the condition numbers
of positive semidefinite matrix sensing?
\par}

The  question concerns the fundamental lower limit of stability. For the rank-one case which corresponds to phase retrieval, these questions have been studied in recent work \cite{XiaXuXu2024}. However, the corresponding problems for general low-rank positive semidefinite matrix sensing remain open. In this paper, we resolve this question by determining the optimal lower bounds and establishing the asymptotic optimality of rank-one Gaussian measurements except for $(\HH,r)=(\R,1)$.

\subsection{Related work}

The rank-one case is of particular interest because of its connection
with phase retrieval, for which sharp bounds on condition numbers
have recently been established.
Xia, Xu, and Xu \cite{XiaXuXu2024} introduced the condition number
$\beta_{\A}^{\mathrm{pr}}=U_{\A}^{\mathrm{pr}}/L_{\A}^{\mathrm{pr}}$
for the magnitude map
\[
\Phi_A^{\mathrm{pr}}(\bm{x})
:=\abs{A\bm{x}}
=\bigl(\abs{\bm{a}_1^*\bm{x}},\dots,\abs{\bm{a}_m^*\bm{x}}\bigr),
\]
where $A$ has rows $\bm{a}_i^*$, and proved the universal lower bound
\begin{equation}
\beta_{\A}^{\mathrm{pr}}\ge
\begin{cases}
\sqrt{\pi/(\pi-2)}, & \HH=\R,\\
\sqrt{4/(4-\pi)}, & \HH=\C.
\end{cases}
\label{eq:phasebound}
\end{equation}
They also showed that these bounds are asymptotically attained by
Gaussian measurements. In the real case, they further obtained
the sharper bound
\[
\beta_{\A}^{\mathrm{pr}}
\geq \left(1-\frac{1}{m\sin(\pi/(2m))}\right)^{-1/2},
\]
which depends explicitly on the number $m$ of measurements.
In \cite{XuXuZhang2026}, a connection was established between the optimal condition number in phase retrieval and the perimeter-maximizing isodiametric problem, originally proposed by Reinhardt in 1922. This connection illustrates the relevance of condition number optimization to geometric extremal problems.

For the intensity map $\Psi_A(\bm{x})=\abs{A\bm{x}}^2$, Peng, Han, and Huang \cite{PengHanHuang2025} defined
\begin{equation}
L^{\ell_p}_{\Psi_A}=\inf_{\bm{x}\bm{x}^* \neq \bm{y}\bm{y}^*} \frac{\norm{\Psi_A(\bm{x})-\Psi_A(\bm{y})}_p}{\norm{\bm{x}\bm{x}^*-\bm{y}\bm{y}^*}_*},
\qquad
U^{\ell_p}_{\Psi_A}=\sup_{\bm{x}\bm{x}^* \neq \bm{y}\bm{y}^*} \frac{\norm{\Psi_A(\bm{x})-\Psi_A(\bm{y})}_p}{\norm{\bm{x}\bm{x}^*-\bm{y}\bm{y}^*}_*},
\label{eq:intensitynorms}
\end{equation}
for $1\le p<\infty$, so that $\beta_{\Psi_A}^{\ell_p}=U^{\ell_p}_{\Psi_A}/L^{\ell_p}_{\Psi_A}$. They proved the universal bounds
\begin{equation}
\beta_{\Psi_A}^{\ell_1}\ge
\begin{cases}
\pi/2, & \HH=\R,\\
2, & \HH=\C,
\end{cases}
\qquad
\beta_{\Psi_A}^{\ell_2}\ge
\begin{cases}
\sqrt{3}, & \HH=\R,\\
2, & \HH=\C,
\end{cases}
\label{eq:intensitybound}
\end{equation}
with the sharper real estimate $\beta_{\Psi_A}^{\ell_1}\ge m\tan(\pi/(2m))$. They also established asymptotic sharpness for Gaussian measurements. In the rank-one case, our nuclear-norm condition number coincides with their $\ell_1$ condition number.

As for general low-rank positive semidefinite matrix sensing, the literature mainly focuses on the stability of the
recovery process rather than the condition number of the sensing map. Let $X_0\in \Snp(\HH)$ be the target matrix. The measurements obtained from the sensing map in \eqref{eq:sensingmap} can be represented as
\[
\bm{y}=\PhiA(X_0)+\bm{w},
\]
where $\bm{w}\in \R^m$ is the noise term. The convex estimators that appear most often in the literature are the following:
\begin{align}
\widehat X_{\mathrm{tr}} &\in
\mathop{\mathrm{arg\,min}}\limits_{X\succeq0} \Tr(X)
\quad\text{subject to}\quad \norm{\PhiA(X)-\bm{y}}_1\le \varepsilon,
\label{eq:tracefit}
\\
\widehat X_{2,+} &\in
\mathop{\mathrm{arg\,min}}\limits_{X\succeq0} \norm{\PhiA(X)-\bm{y}}_2,
\label{eq:squarefit}
\\
\widehat X_{1,+} &\in
\mathop{\mathrm{arg\,min}}\limits_{X\succeq0} \norm{\PhiA(X)-\bm{y}}_1.
\label{eq:robustfit}
\end{align}
The program \eqref{eq:tracefit} uses the trace as a convex surrogate for the rank, whereas \eqref{eq:squarefit} and \eqref{eq:robustfit} are regularization-free and rely only on the positive semidefinite constraint. A commonly used sensing map in previous work is the rank-one Gaussian map $A_i=\bm{a}_i \bm{a}_i^*$, where $\bm{a}_i$ are i.i.d. standard Gaussian vectors. Stability under the rank-one Gaussian measurement ensemble has been established for all three
estimators.
Chen, Chi, and Goldsmith~\cite{ChenChiGoldsmith2015} established
stability guarantees for the trace-minimization estimator
$\widehat X_{\mathrm{tr}}$ in \eqref{eq:tracefit},
while Kabanava et al.~\cite{KabanavaKuengRauhutTerstiege2016}
obtained stability guarantees for the PSD least-squares estimator
in \eqref{eq:squarefit}.
Li, Sun, and Chi~\cite{LiSunChi2017} analyzed the $\ell_1$ residual
model \eqref{eq:robustfit}.
They consider corrupted measurements where the noise term $\bm{w}=
\bm{\eta}+\bm{w}'$ with $\norm{\bm{w}'}_1\le \varepsilon$ and $\norm{\bm{\eta}}_0=sm$, i.e. the fraction $s\in [0,1/2)$ of measurements may be arbitrarily corrupted. Suppose that $X_0$ is a fixed rank-$r$ PSD matrix and the support of $\bm{\eta}$ is selected uniformly at random with the signs of its nonzero entries generated from the Rademacher distribution. If $m\ge c_1nr^2$ and $s\le {s_0}/r$, then with high probability,
\[
\norm{\widehat X_{1,+}-X_0}_F\le c_2\frac{r\varepsilon}{m}.
\]
Here $c_1, c_2$, and $s_0$ are universal constants. Note that this theorem holds for fixed $X_0$ rather than for all target matrices simultaneously.

\subsection{Notation}
Throughout this paper, we denote the entries of a vector $\bm{x}$ by $x_i$. We say a vector $\bm{a}\in\HH^d$ is a standard Gaussian random vector if $a_i$ are
i.i.d.\ $\mathcal N(0,1)$ for $\HH=\R$, and $a_i$ are i.i.d.\ $\mathcal N(0,1/2)+ \mathrm{i}\mathcal N(0,1/2)$ for $\HH=\C$.
We use $\Gamma(p,\theta)$ to denote the Gamma distribution with shape parameter $p$ and scale parameter $\theta$, and use $\operatorname{Beta}(p,q)$ to denote the Beta distribution with shape parameters $p$ and $q$. We use $\chi^2(r)$ to denote the chi-square distribution with parameter $r$. For a random variable $X$, we write $\norm{X}_{L_p}:=(\E \abs{X}^p)^{1/p}$. For random variables $X$ and $Y$, we write $X\overset{d}{=}Y$ if they have the same distribution.

\begin{table}[ht]
\centering
\caption{Summary of notation used throughout the paper.}
\label{tab:symbolguide}
\renewcommand{\arraystretch}{1.15}
\begin{tabular}{@{}p{0.30\textwidth}p{0.66\textwidth}@{}}
\hline
Symbol & Definition \\
\hline
$\HH$ & The underlying field, either $\mathbb{R}$ or $\mathbb{C}$.\\
$\A$
& The collection of measurement matrices. \\
$\Aone$
& The rank-one Gaussian measurement ensemble. \\

$\Sn(\HH)$
& The space of $n\times n$ symmetric matrices over $\R$ or Hermitian matrices over $\C$. \\

$\Snp(\HH)$
& The cone of positive semidefinite matrices in $\Sn(\HH)$. \\

$\Mr(\HH)$
& The set of matrices in $\Snp(\HH)$ with rank at most $r$. \\

$\beta_{\A}^{\HH,*}(r)$ and $\beta_{\A}^{\HH,2}(r)$
& Condition numbers associated with the measurement ensemble $\A$, defined in \eqref{eq:conditionratio}. \\

$\beta_{0}^{\HH,*}(r)$ and $\beta_{0}^{\HH,2}(r)$
& Uniform bounds for $\beta_{\A}^{\HH,*}(r)$ and $\beta_{\A}^{\HH,2}(r)$, respectively, defined in Definition~\ref{def:betaconstants}. \\

$\eta^\HH(r,t)$
& An expectation involving a Beta-distributed random variable, used in the definition of $\beta_{0}^{\HH,2}(r)$; see Definition~\ref{def:betaconstants}. \\

$\eta_0^\HH(r)$
& The minimum value of $\eta^\HH(r,t)$ over $t\in[0,1]$, namely,
$\eta_0^\HH(r):=\inf_{0\leq t\leq 1}\eta^\HH(r,t)$. \\

$\prec$
& The majorization relation; see Lemma~\ref{lem:convexorder}. \\

$\mu^{\HH,*}(r)$ and $\mu^{\HH,2}(r)$
& Stability constants for the PhaseLift model, defined in Definition~\ref{def:meanminimum}. \\

$\mu^{\HH,2}$
& An abbreviation for $\mu^{\HH,2}(r)$; see Lemma~\ref{lem:meanvalues}. \\
\hline
\end{tabular}
\end{table}

\subsection{Our contribution}

\subsubsection{Universal Lower Bounds for Condition Numbers}

In this paper we obtain asymptotically tight universal lower bounds for the condition numbers $\beta_{\A}^{\HH,*}(r)$ and $\beta_{\A}^{\HH,2}(r)$, valid for every $\A=\{A_i\}_{i=1}^m\subset \Snp(\HH)$.
\begin{definition}\label{def:betaconstants}
For each integer $r\geq 1$, we define the universal constants $\beta_0^{\HH,*}(r)$ and $\beta_0^{\HH,2}(r)$ as follows.
\[
\beta_0^{\R,*}(r)
:=
\frac{r\sqrt{\pi}\Gamma(r/2)}{2\Gamma\left((r+1)/2\right)},
\qquad
\beta_0^{\C,*}(r)
:=
\frac{r\sqrt{\pi}\Gamma(r)}
{\Gamma\left((2r+1)/2\right)},
\]
and
\[
\beta_0^{\HH,2}(r)
:=
\frac{r}{(r+1)\eta_0^\HH(r)},
\]
where
\(
\eta_0^\HH(r):=\inf_{0\leq t\leq1}\eta^\HH(r,t).
\)
For $\HH=\R$,
\[
\eta^\R(r,t)
:=
\frac{\Gamma\left((r+1)/2\right)}{\Gamma(1/2)\Gamma(r/2)}
\int_0^1
\abs{(1+t)u-t}
u^{-1/2}(1-u)^{r/2-1}\,du,
\]
and for $\HH=\C$,
\[
\eta^\C(r,t)
:=
r\int_0^1
\abs{(1+t)u-t}(1-u)^{r-1}\,du
=
\frac{rt-1}{r+1}
+\frac{2}{(r+1)(1+t)^r}.
\]
In particular,
\(
\eta_0^\C(r)=2^{1/(r+1)}-1,
\)
and hence
\(
\beta_0^{\C,2}(r)
=
\frac{r}{(r+1)\left(2^{1/(r+1)}-1\right)}.
\)
\end{definition}

The following theorem shows that $\beta_0^{\HH,*}(r)$ and $\beta_0^{\HH,2}(r)$ provide universal lower bounds for $\beta_{\A}^{\HH,*}(r)$ and $\beta_{\A}^{\HH,2}(r)$, respectively.
\begin{theorem}\label{thm:lowerbound}
Let $\A=\{A_i\}_{i=1}^m$, where $A_i\in\Snp(\HH)$. Suppose that $1\leq r\le n/2$. Then we have
\[
\beta_{\A}^{\HH,*}(r) \ge \beta_0^{\HH,*}(r) \quad \text{and} \quad \beta_{\A}^{\HH,2}(r) \ge \beta_0^{\HH,2}(r).
\]
\end{theorem}

We present the proof of Theorem~\ref{thm:lowerbound} in Section 2.  To our knowledge, this result provides the first universal constant lower bounds on the condition numbers for low-rank positive semidefinite matrix sensing. Previous studies mainly show that certain random ensembles or specific measurement classes guarantee stable recovery. However, they do not quantify the precise condition number  achievable by positive semidefinite sensing maps.  Moreover, in the rank-one case, Theorem~\ref{thm:lowerbound}
recovers the lower bounds for intensity phase retrieval in \eqref{eq:intensitybound}, since $\beta_0^{\R,*}(1)=\pi/2, \beta_0^{\C,*}(1)=2$.

Additionally, we show that the lower bounds in Theorem~\ref{thm:lowerbound} are asymptotically tight, up to one exceptional case, for both the nuclear norm and operator norm condition numbers. We prove this by considering rank-one Gaussian measurements.

\begin{theorem}\label{thm:randombounds}
Let $\bm{a}_1,\ldots,\bm{a}_m\in\HH^n$ be independent standard Gaussian
vectors. Consider the rank-one Gaussian measurement ensemble
$
\Aone:=\{\bm{a}_i\bm{a}_i^*\}_{i=1}^m.
$
Suppose that $1\leq r\leq n/2$. Then, for every $\delta\in(0,1)$ satisfying $m\geq C\delta^{-2}nr$, with probability at least
$1-2\exp(-c\delta^2m)$,
\begin{equation}\label{eq:nuclearsharp}
\beta_0^{\HH,*}(r)
\leq
\beta_{\Aone}^{\HH,*}(r)
\leq
\beta_0^{\HH,*}(r)\frac{1+\delta}{1-\delta},
\end{equation}
and
\begin{equation}\label{eq:operatorsharp}
\beta_0^{\HH,2}(r)
\leq
\beta_{\Aone}^{\HH,2}(r)
\leq
\frac{U_r^{\HH}}{(r+1)\eta_0^{\HH}(r)}
\frac{1+\delta}{1-\delta}.
\end{equation}
Here $c,C>0$ are universal constants, and
\[
U_r^{\HH}
:=
\begin{cases}
\dfrac{4}{\pi}, & \text{if }(\HH,r)=(\R,1),\\[2mm]
r, & \text{otherwise}.
\end{cases}
\]
\end{theorem}

We prove Theorem~\ref{thm:randombounds} in Section 3. Recall that $\beta_0^{\HH,2}(r):=r/((r+1)\eta_0^{\HH}(r))$. By letting $\delta \to 0$ in Theorem~\ref{thm:randombounds}, the condition numbers $\beta_{\Aone}^{\HH,*}(r)$ and $\beta_{\Aone}^{\HH,2}(r)$ approach $\beta_{0}^{\HH,*}(r)$ and $\beta_{0}^{\HH,2}(r)$ respectively, except for $\HH=\R$ and $r=1$.
Thus, Theorem~\ref{thm:randombounds} shows that the universal lower bounds
$\beta_{0}^{\HH,*}(r)$ and $\beta_{0}^{\HH,2}(r)$ are asymptotically
attained, except for the case $(\HH,r)=(\R,1)$, which will be treated
separately in the following proposition.
\begin{proposition}\label{prop:rankone}
Let $\vartheta_0$ denote the unique $\vartheta\in(0,\pi/2)$ satisfying $\cos(\vartheta)-\vartheta=2-\pi/2$.
 Let $\A=\{A_i\}_{i=1}^m$, where $A_i\in\Snp(\R)$ and $n\ge 2$. Then we have the lower bound
\[
\beta_{\mathcal A}^{\R,2}(1) \,\,\ge\,\, \sec{\vartheta_0}\,\, \approx\,\, 1.1186435067,
\]
which is asymptotically sharp.
\end{proposition}
\noindent
The proof of Proposition~\ref{prop:rankone} is presented in the Appendix.

\subsubsection{Recovery Guarantees for the Model in \eqref{eq:robustfit}}

We next establish recovery guarantees for \eqref{eq:robustfit} under random rank-one Gaussian measurements. The condition numbers above concern perturbations $H=X-Y$ with $X,Y\in\mathcal M_r(\HH)$. So $H$ has at most $r$ positive and $r$ negative eigenvalues. By contrast, the recovery error $D=X_0-\widehat X_{1,+}$ for  \eqref{eq:robustfit} may have full rank but satisfies $\nu_+(D)\le r$, where $\nu_+(D)$ denotes the number of positive eigenvalues of $D$, counted with multiplicity. The constants introduced below extend the corresponding lower Lipschitz constants to this larger one-sided perturbation class.

\begin{definition}\label{def:meanminimum}
Let $\bm{a}\in\HH^n$ be a standard Gaussian vector. Define
\begin{equation}\label{eq:meanminimum}
\mu^{\HH,*}(r):=\inf_{n\geq r} \inf_{H\in T^*(n,r)} \E \abs{\bm{a}^*H\bm{a}} \text{\ \ and\ \ } \mu^{\HH,2}(r):=\inf_{n\ge r} \inf_{H\in T^2(n,r)} \E \abs{\bm{a}^*H\bm{a}},
\end{equation}
\[
T^*(n, r) :=\left\{H\in\Sn(\HH):\norm{H}_*=1,\ \nu_+(H)\le r\right\},
\]
 and
 \[
 T^2(n, r) :=\left\{H\in\Sn(\HH):\norm{H}_2=1,\ \nu_+(H)\le r\right\},
 \]
where $\nu_+(H)$ denotes the number of positive eigenvalues of $H$ counted with multiplicity.
\end{definition}

Thus, $\mu^{\HH,*}(r)$ and $\mu^{\HH,2}(r)$ are not separate from the condition number framework developed above. Rather, they extend its lower stability component to the perturbation class arising from \eqref{eq:robustfit}. The following Lemma~\ref{lem:meanvalues} makes this connection quantitative by relating these constants asymptotically to
$\beta_0^{\HH,*}(r)$ and $\beta_0^{\HH,2}(r)$.

\begin{lemma}\label{lem:meanvalues}
Let $m_1$ and $m_2$ denote the medians of $\chi^2(1)$ and
$\Gamma(1,1)$, respectively. Then, for all $r\geq1$,
\[
\mu^{\R,2}(r)=\E\abs{\chi^2(1)-m_1}\,\approx\,0.857348,
\qquad
\mu^{\C,2}(r)=\E\abs{\Gamma(1,1)-m_2}=\ln 2.
\]
In particular, $\mu^{\HH,2}(r)$ is independent of $r$, and we henceforth write it simply as $\mu^{\HH,2}$. Furthermore, the following asymptotic relations hold:
\[
\lim_{r\to\infty}\mu^{\HH,*}(r)\beta_0^{\HH,*}(r)
=
\frac{1}{\sqrt{2}},
\]
and
\[
\lim_{r\to\infty}(r+1)\eta_0^{\HH}(r)
=
\mu^{\HH,2},
\]
where the quantities $\beta_0^{\HH,*}(r)$ and $\eta_0^{\HH}(r)$ are defined in Definition \ref{def:betaconstants}.
\end{lemma}

The proof of Lemma~\ref{lem:meanvalues} is given in Section~4. The following theorem establishes stability guarantees for solutions of \eqref{eq:robustfit}.

\begin{theorem}\label{thm:robustrecovery}
Let $\bm{w},\bm{\eta}\in \R^m$ satisfy $\norm{\bm{\eta}}_0= s\cdot m$ with $s\in [0,1/2)$. Suppose $\hat{X}_{1,+}$ is a solution of \eqref{eq:robustfit} with $\bm{y}=\Phi_{\Aone}(X_0)+\bm{w}+\bm{\eta}$, where $\Aone:=\{\bm{a}_i\bm{a}_i^*\}_{i=1}^m$ is the rank-one Gaussian measurement ensemble. Then, for all $X_0 \in\mathcal M_r(\HH)$ and every $\delta\in(0,1)$ satisfying $m \ge C\delta^{-2}(1-2s)^{-2}nr$, with probability at least $1-4\exp(-c\delta^2(1-2s)^2 m)$,
\[
\norm{X_0-\hat{X}_{1,+}}_*\le\frac{2}{(1-\delta)(1-2s)\mu^{\HH,*}(r)}\cdot\frac{\norm{\bm{w}}_1}{m}
\]
and
\[
\norm{X_0-\hat{X}_{1,+}}_2\le\frac{2}{(1-\delta)(1-2s)\mu^{\HH,2}}\cdot\frac{\norm{\bm{w}}_1}{m}.
\]
Here $c,C>0$ are universal constants.
\end{theorem}

We prove Theorem~\ref{thm:robustrecovery} in Section~4.

\begin{remark}
 Compared with \cite[Theorem~1]{LiSunChi2017}, our result provides uniform recovery of all $X_0\in\mathcal M_r(\HH)$, for both $\HH=\R$ and $\HH=\C$, at the optimal sampling rate $m\asymp nr$. It applies to every fixed corruption vector $\bm{\eta}$ with arbitrary support and signs whenever $s<1/2$. In contrast, \cite{LiSunChi2017} treats a fixed real target matrix, assumes uniformly random corruption support and independent Rademacher signs, requires $m\asymp nr^2$ and $s\le s_0/r$, and provides a Frobenius-norm error bound with an implicit constant. Our nuclear- and operator-norm bounds have explicit robustness constants expressed through $\mu^{\HH,*}(r)$ and $\mu^{\HH,2}$, which Lemma~\ref{lem:meanvalues} relates directly to the universal condition number bounds in Theorem~\ref{thm:lowerbound}.
\end{remark}

\section{Proof of Theorem~\ref{thm:lowerbound}: Universal Lower Bounds for Condition Numbers}

To derive lower bounds for the condition numbers
$\beta_{\A}^{\HH,*}(r)$ and $\beta_{\A}^{\HH,2}(r)$, we first establish the
following two lemmas, which characterize the upper and lower Lipschitz
constants separately. Their proofs are deferred to the end of this
section.

\begin{lemma}\label{lem:upperconstant}
Let $\A=\{A_i\}_{i=1}^m$, where $A_i\in\Snp(\HH)$. Suppose $\lambda_1\ge \lambda_2\ge \cdots \ge \lambda_n\ge 0$ are the eigenvalues of $S_{\mathcal A}:=\sum_{i=1}^m A_i$. Then we have
\begin{equation}\label{eq:upperconstant}
U_{\A}^{\HH,*}(r)=\lambda_1\quad \text{and}\quad U_{\A}^{\HH,2}(r)\ge \sum_{j=1}^{r}\lambda_j.
\end{equation}
\end{lemma}

\begin{lemma}\label{lem:lowerconstant}
Let $1\leq r\le n/2$ and $\A=\{A_i\}_{i=1}^m$, where $A_i\in\Snp(\HH)$. Suppose $\lambda_1\ge \lambda_2\ge \cdots \ge \lambda_n\ge 0$ are the eigenvalues of $S_{\mathcal A}:=\sum_{i=1}^m A_i$. Then we have
\begin{equation}\label{eq:lowerconstant}
L_{\A}^{\HH,*}(r) \le \frac{1}{2r\beta_0^{\HH,*}(r)} \sum_{j=1}^{2r}\lambda_j\quad \text{and}\quad L_{\A}^{\HH,2}(r)\le\eta_0^{\HH}(r)\sum_{j=1}^{r+1}\lambda_j.
\end{equation}
Here, the quantities $\beta_0^{\HH,*}(r)$ and $\eta_0^{\HH}(r)$ are defined in Definition \ref{def:betaconstants}.

\end{lemma}
We now prove Theorem~\ref{thm:lowerbound} using Lemma~\ref{lem:upperconstant} and Lemma~\ref{lem:lowerconstant}.
\begin{proof}[Proof of Theorem~\ref{thm:lowerbound}]
Following the notation introduced in Lemma~\ref{lem:upperconstant} and Lemma~\ref{lem:lowerconstant}, let $\lambda_1\ge \lambda_2\ge \cdots \ge \lambda_n\ge 0$ denote the eigenvalues of $S_{\mathcal A}:=\sum_{i=1}^m A_i$. Combining Lemma~\ref{lem:upperconstant} and Lemma~\ref{lem:lowerconstant}, we have
\[
\beta_{\A}^{\HH,*}(r)
=\frac{U_{\A}^{\HH,*}(r)}{L_{\A}^{\HH,*}(r)}
\ge \frac{2r\lambda_1}{\sum_{j=1}^{2r}\lambda_j}\cdot \beta_0^{\HH,*}(r)\ge \beta_0^{\HH,*}(r),
\]
\[
\beta_{\A}^{\HH,2}(r)=\frac{U_{\A}^{\HH,2}(r)}{L_{\A}^{\HH,2}(r)}\ge \frac{\sum_{j=1}^{r} \lambda_j}{\eta_0^{\HH}(r)\sum_{j=1}^{r+1}\lambda_j}\ge \frac{r}{(r+1)\eta_0^{\HH}(r)}=\beta_0^{\HH,2}(r).
\]
This completes the proof.
\end{proof}

We next present the proof of  Lemma \ref{lem:upperconstant}.
\begin{proof}[Proof of Lemma \ref{lem:upperconstant}]
Let $X,Y\in \Mr$ and define $H:=X-Y$. Since $\Phi_{\mathcal A}$ is linear,
we have
\[
\Phi_{\mathcal A}(X)-\Phi_{\mathcal A}(Y)
=
(\Tr(A_1H),\dots,\Tr(A_mH)).
\]

We first consider the nuclear norm case. Write the spectral decomposition
of $H$ as $H=U^*DU$, where $D$ is a real diagonal matrix and $U$ is unitary.
Define $\widetilde H:=U^*\abs{D}U$, where $\abs{D}$ denotes the diagonal
matrix obtained by taking the entrywise absolute values of the diagonal
entries of $D$. Then we have
$
-\widetilde H\preceq H\preceq \widetilde H.
$
Since each $A_i$ is positive semidefinite, it follows that
$
\abs{\Tr(A_iH)}\leq \Tr(A_i\widetilde H),
 i=1,\dots,m.
$
Therefore, summing over all measurements gives
\begin{equation*}
\norm{\Phi_{\mathcal A}(X)-\Phi_{\mathcal A}(Y)}_1=\norm{\Phi_{\mathcal A}(H)}_1
\leq
\Tr\left(S_{\mathcal A}\widetilde H\right)
\leq
\lambda_1\Tr\left(\widetilde H\right)
=
\lambda_1\norm{H}_*
=
\lambda_1\norm{X-Y}_*.
\end{equation*}
Hence, by the definition of $U_{\A}^{\HH,*}(r)$ in (\ref{eq:nuclearbounds}), we obtain
\begin{equation}\label{eq:upperbound}
U_{\A}^{\HH,*}(r)\,\,\leq\,\, \lambda_1.
\end{equation}

It remains to show that the above bound is attainable. Let
$\bm{u}_1,\dots,\bm{u}_r$ be orthonormal eigenvectors of
$S_{\mathcal A}$ corresponding to the eigenvalues
$\lambda_1,\dots,\lambda_r$, respectively. Choosing
$X=\bm{u}_1\bm{u}_1^*$ and $Y=0$, we have
\begin{equation}\label{eq:upperequality}
\frac{\norm{\Phi_{\mathcal A}(X)-\Phi_{\mathcal A}(Y)}_1}
{\norm{X-Y}_*}
=
\bm{u}_1^*S_{\mathcal A}\bm{u}_1
=
\lambda_1.
\end{equation}
Combining \eqref{eq:upperbound} and \eqref{eq:upperequality} yields
$
U_{\A}^{\HH,*}(r)=\lambda_1.
$

We next establish the lower bound for $U_{\A}^{\HH,2}(r)$. Let
$V=(\bm{u}_1,\dots,\bm{u}_r)$ and choose
$X=VV^*$ and $Y=0$. Since the columns of $V$ are orthonormal,
$\norm{X-Y}_2=1$. Moreover,
\begin{equation}\label{eq:operatorsum}
\frac{\norm{\Phi_{\mathcal A}(X)-\Phi_{\mathcal A}(Y)}_1}
{\norm{X-Y}_2}
=
\sum_{j=1}^{r}\lambda_j.
\end{equation}
Therefore, by the definition of $U_{\A}^{\HH,2}(r)$, we conclude that
\(
U_{\A}^{\HH,2}(r)\geq \sum_{j=1}^{r}\lambda_j.
\)
This completes the proof.
\end{proof}

It remains to prove Lemma~\ref{lem:lowerconstant}. The guiding idea here is that we evaluate $\Phi_{\mathcal A}$ on a carefully chosen family of low-rank test matrices and then average over a Haar-distributed parameter. We start with a technical lemma about Beta and Gamma distributions, and then establish the Haar-averaging lemma required in the proof.

\begin{lemma}\label{lem:gammaratio}\textnormal{\cite[Example~2.41]{Ross2024}}
Let $p, q, \theta$ be positive constants. Suppose that $G_1\sim \Gamma(p,\theta)$ and $G_2\sim \Gamma(q,\theta)$ are independent random variables.
Define
\[
S:=G_1+G_2,\qquad U:=\frac{G_1}{G_1+G_2}.
\]
Then $S$ and $U$ are independent, and
\[
S\sim \Gamma(p+q,\theta),
\qquad
U\sim \mathrm{Beta}(p,q).
\]
\end{lemma}
The following lemma gives the averaging estimates used in the proof of  Lemma~\ref{lem:lowerconstant}.

\begin{lemma}\label{lem:rotationmean}

Let $Q$ be Haar distributed on the set of orthogonal matrices when $\HH=\R$, and on the set of unitary matrices when
$\HH=\C$. We have the following two properties.
\begin{enumerate}[label=(\roman*)]
\item Set $D:= \diag(I_r,-I_r)\in\R^{2r\times 2r}$. Then for any $B\in \mathbb S_{2r}^+(\HH)$,
\begin{equation}\label{eq:nuclearaverage}
\E_Q \abs{\Tr(BQDQ^*)}
\le
\frac{1}{\beta_0^{\HH,*}(r)}\Tr(B).
\end{equation}
\item Set $D_t:=\diag(1,-t,\dots,-t)\in \R^{(r+1)\times(r+1)}$, where $0\le t\le 1$. Then for any $B\in \mathbb S_{r+1}^+(\HH)$,
\begin{equation}\label{eq:operatoraverage}
\E_Q \abs{\Tr(BQD_tQ^*)}
\le
\eta^{\HH}(r,t)\Tr(B).
\end{equation}
\end{enumerate}
\end{lemma}
\begin{proof}
Let $B\in \mathbb S_{2r}^+(\HH)$ in case (i) and $B\in \mathbb S_{r+1}^+(\HH)$ in case (ii). In either case, write the spectral decomposition of $B$ as
\[
B=\sum_{\ell}\mu_\ell \bm{z}_\ell \bm{z}_\ell^*,
\]
where $\mu_\ell\ge0$ are the eigenvalues and $\bm{z}_\ell$ are the corresponding orthonormal eigenvectors.
 For either diagonal matrix under consideration, denoted generically by $D$, we have
\begin{equation}\label{eq:spectralmean}
\E_Q\abs{\Tr(BQDQ^*)}
\le
\E_Q\sum_{\ell}\mu_\ell \abs{\bm{z}_\ell^*QDQ^*\bm{z}_\ell}
=\sum_{\ell}\mu_\ell \E_Q\abs{\bm{z}_\ell^*QDQ^*\bm{z}_\ell}.
\end{equation}
Since $Q$ is Haar-distributed and $\bm{z}_\ell$ is a fixed unit vector, the random vector $\bm{q}:=Q^*\bm{z}_\ell$ is uniformly distributed on the unit sphere. Suppose $\bm{g}$ is a standard Gaussian random vector. Since a normalized Gaussian vector is uniformly distributed on the unit sphere, we have $\bm{q}\stackrel{d}{=} \bm{g}/\norm{\bm{g}}_2$.

We first prove \eqref{eq:nuclearaverage}. Note that
\begin{equation}\label{eq:signedmass}
\bm{q}^*D\bm{q}=\sum_{i=1}^r \abs{q_i}^2-\sum_{i=r+1}^{2r}\abs{q_i}^2=2\sum_{i=1}^r \abs{q_i}^2-1.
\end{equation}
We claim that
\begin{equation}\label{eq:betamean}
\E\abs{2\sum_{i=1}^r \abs{q_i}^2-1}= \frac{1}{\beta_0^{\HH,*}(r)}.
\end{equation}
Once the above equation is proved, \eqref{eq:spectralmean} and \eqref{eq:signedmass} give
\[
\E_Q\abs{\Tr(BQDQ^*)}
\le
\sum_{\ell=1}^{2r}\mu_\ell \E\abs{\bm{q}^*D\bm{q}}
=
 \frac{1}{\beta_0^{\HH,*}(r)}\sum_{\ell=1}^{2r}\mu_\ell
=
\frac{1}{\beta_0^{\HH,*}(r)}\Tr(B).
\]
It remains to prove \eqref{eq:betamean}. If $\HH=\C$, then $\sum_{i=1}^r \abs{g_i}^2$ and $\sum_{i=r+1}^{2r}\abs{g_i}^2$ are independent $\Gamma(r,1)$ random variables. If $\HH=\R$, then $\sum_{i=1}^r \abs{g_i}^2$ and $\sum_{i=r+1}^{2r}\abs{g_i}^2$ are independent $\Gamma(r/2,2)$ random variables. Hence, Lemma~\ref{lem:gammaratio} gives
\[
\sum_{i=1}^r |q_i|^2\sim \mathrm{Beta}(c_{\HH}\cdot r,c_{\HH}\cdot r),
\]
where $c_{\HH}=1$ if $\HH=\C$ and $c_{\HH}=1/2$ if $\HH=\R$. A direct calculation gives
\[
\begin{split}
\E\abs{2\sum_{i=1}^r \abs{q_i}^2-1}&=\frac{2\Gamma(2c_{\HH}r)}{\Gamma(c_{\HH}r)^2}\int_{1/2}^{1}(2u-1)u^{c_{\HH}r-1}(1-u)^{c_{\HH}r-1}\,du \\
&=
\frac{2^{1-2c_{\HH}r}\Gamma(2c_{\HH}r)}{c_{\HH}r\Gamma(c_{\HH}r)^2}
=
\frac{\Gamma\left(c_{\HH}r+1/2\right)}{c_{\HH}r\sqrt{\pi}\Gamma(c_{\HH}r)}= \frac{1}{\beta_0^{\HH,*}(r)}.
\end{split}
\]

Next we prove \eqref{eq:operatoraverage}. Note that
\begin{equation}\label{eq:weightedsquare}
\abs{\bm{q}^*D_t\bm{q}}
=
\abs{\abs{q_1}^2-t\sum_{i=2}^{r+1}\abs{q_i}^2}
=
\abs{(1+t)\abs{q_1}^2-t}.
\end{equation}
If $\HH=\C$, then $|g_1|^2\sim \Gamma(1,1)$ and $\sum_{i=2}^{r+1}|g_i|^2\sim \Gamma(r,1)$ are independent. If $\HH=\R$, then $|g_1|^2\sim \Gamma(1/2,2)$ and $\sum_{i=2}^{r+1}|g_i|^2\sim \Gamma(r/2,2)$ are independent. Hence, Lemma~\ref{lem:gammaratio} gives
\[
|q_1|^2\sim \mathrm{Beta}(c_{\HH},c_{\HH}\cdot r),
\]
where $c_{\HH}=1$ when $\HH=\C$ and $c_{\HH}=1/2$ when $\HH=\R$. Thus by the definition of $\eta^{\HH}(r,t)$,
\[
\E\abs{(1+t)|q_1|^2-t}
=
\eta^{\HH}(r,t).
\]
Combining \eqref{eq:spectralmean} and \eqref{eq:weightedsquare}, we have
\[
\E_Q \abs{\Tr(BQD_tQ^*)}
\le
\sum_{\ell=1}^{r+1}\mu_\ell \E\abs{\bm{q}^*D_t\bm{q}}
=
\eta^{\HH}(r,t)\sum_{\ell=1}^{r+1}\mu_\ell
=
\eta^{\HH}(r,t)\Tr(B).
\]
This completes the proof.
\end{proof}

\begin{proof}[Proof of Lemma~\ref{lem:lowerconstant}]
We first prove the upper bound for $L_{\mathcal A}^{\mathbb H,*}(r)$. Let $Q$ be Haar distributed on the set of orthogonal matrices $O(2r)$ if $\HH=\R$, and on the set of unitary matrices $U(2r)$ if $\HH=\C$. Define
\[
E_X:=\begin{pmatrix}
I_r & 0\\
0 & 0
\end{pmatrix}, \qquad
E_Y:=\begin{pmatrix}
0 & 0\\
0 & I_r
\end{pmatrix}, \qquad
J_r:=\begin{pmatrix}
I_r & 0\\
0 & -I_r
\end{pmatrix}.
\]
Set
\begin{equation}\label{eq:testmatrices}
X_Q
:=
\begin{pmatrix}
QE_XQ^* & 0\\
0 & 0
\end{pmatrix},
\qquad
Y_Q
:=
\begin{pmatrix}
QE_YQ^* & 0\\
0 & 0
\end{pmatrix}.
\end{equation}
We have $X_Q,Y_Q\in \mathcal M_r(\HH)$ and $\norm{X_Q-Y_Q}_*=2r$. Therefore, by the definition of $L_{\A}^{\HH,*}(r)$ in \eqref{eq:nuclearbounds}, we obtain
\begin{equation}\label{eq:nuclearstart}
L_{\A}^{\HH,*}(r)
\le
\frac{1}{2r}\E_Q \norm{\Phi_{\mathcal A}(X_Q)-\Phi_{\mathcal A}(Y_Q)}_1
=
\frac{1}{2r}\sum_{i=1}^m
\E_Q
\abs{\Tr (A_i(X_Q-Y_Q) )}.
\end{equation}
Set $P_*:=( I_{2r},0 )\in \R^{2r\times n}$. Let $B_i:=P_*A_iP_*^\top\in \HH^{2r\times 2r}$ denote the corresponding principal submatrix. Since the support of $X_Q-Y_Q$ is contained in the leading principal block, we have
\begin{equation}\label{eq:tracerewrite}
\Tr (A_i(X_Q-Y_Q) )=\Tr(B_iQJ_rQ^*).
\end{equation}
Applying Lemma~\ref{lem:rotationmean} to each $B_i$ and returning to \eqref{eq:nuclearstart}, we obtain
\begin{equation}\label{eq:nucleartrace}
L_{\A}^{\HH,*}(r)\le \frac{1}{2r}\sum_{i=1}^m\E_Q\abs{\Tr(B_iQJ_rQ^*)}
\le\frac{1}{2r\beta_0^{\HH,*}(r)}
\sum_{i=1}^m \Tr(B_i).
\end{equation}
The maximum principle of normal matrices gives
\begin{equation}\label{eq:tracebound}
\sum_{i=1}^m \Tr(B_i)=\Tr\left(P_*S_{\mathcal A}P_*^\top\right)
\le
\sum_{j=1}^{2r}\lambda_j.
\end{equation}
Combining \eqref{eq:nucleartrace} and \eqref{eq:tracebound} gives
\[
L_{\A}^{\HH,*}(r)
\le
\frac{1}{2r\beta_0^{\HH,*}(r)}\sum_{j=1}^{2r}\lambda_j.
\]

We next establish the upper bound for $L_{\mathcal A}^{\mathbb H,2}(r)$. The analysis is similar to the proof concerning $L_{\A}^{\HH,*}(r)$. Let $Q$ be Haar distributed on the set of orthogonal matrices $O(r+1)$ if $\HH=\R$, and on the set of unitary matrices $U(r+1)$ if $\HH=\C$. Define
\[
E:=\begin{pmatrix}
1 &        &        &   \\
  & 0      &        &   \\
  &        & \ddots &   \\
  &        &        & 0
\end{pmatrix}, \qquad
E_t:=\begin{pmatrix}
0 &        &        &   \\
  & t      &        &   \\
  &        & \ddots &   \\
  &        &        & t
\end{pmatrix}, \qquad
J_{r,t}:=\begin{pmatrix}
1 &        &        &   \\
  & -t      &        &   \\
  &        & \ddots &   \\
  &        &        & -t
\end{pmatrix}.
\]
Let
\[
X_{Q,t}
:=
\begin{pmatrix}
QEQ^* & 0\\
0 & 0
\end{pmatrix},
\qquad
Y_{Q,t}
:=
\begin{pmatrix}
QE_tQ^* & 0\\
0 & 0
\end{pmatrix}.
\]
We have $X_{Q,t},Y_{Q,t}\in \Mr(\HH)$ and $\norm{X_{Q,t}-Y_{Q,t}}_2=1$.
Let $P_2:=( I_{r+1},0 )\in \R^{(r+1)\times n}$ and $B_i:=P_2A_iP_2^\top\in \HH^{(r+1)\times (r+1)}$. By the definition of $L_{\A}^{\HH,2}(r)$ in \eqref{eq:operatorbounds}, we obtain
\[
L_{\A}^{\HH,2}(r)\le\E_Q \norm{\Phi_{\mathcal A}(X_{Q,t})-\Phi_{\mathcal A}(Y_{Q,t})}_1
=\sum_{i=1}^m \E_Q \abs{\Tr(B_iQJ_{r,t}Q^*)}.
\]
Lemma~\ref{lem:rotationmean} gives $\E_Q \abs{\Tr(B_iQJ_{r,t}Q^*)}\le\eta^{\HH}(r,t)\Tr(B_i)$. Thus
\begin{equation}\label{eq:operatortrace}
L_{\A}^{\HH,2}(r)
\le
\eta^{\HH}(r,t)\sum_{i=1}^m \Tr(B_i)
\le \eta^{\HH}(r,t)\sum_{j=1}^{r+1}\lambda_j.
\end{equation}
Since $t\in[0,1]$ is arbitrary, taking the infimum over $t$ gives
\[
L_{\A}^{\HH,2}(r)\le\inf_{t\in[0,1]}\eta^{\HH}(r,t)\sum_{j=1}^{r+1}\lambda_j=\eta_0^{\HH}(r)\sum_{j=1}^{r+1}\lambda_j.
\]
This completes the proof.
\end{proof}

\section{Asymptotic Attainment of the Universal Lower Bounds: Proof of Theorem~\ref{thm:randombounds}}
\subsection{Proof of Theorem~\ref{thm:randombounds}}
To prove Theorem~\ref{thm:randombounds}, we establish two auxiliary results: Lemma~\ref{lem:sampleerror} provides the required concentration inequality, and Lemma~\ref{lem:meanbounds} bounds the expectation used in our analysis. The proofs of both lemmas are deferred to the end of this section.

\begin{lemma}\label{lem:sampleerror}
Let $\bm{a}_1,\ldots,\bm{a}_m\in\HH^n$ be independent standard Gaussian vectors, and $T\subset \Sn(\HH)$. Suppose that there exist a constant $K>0$ and an integer $1\le r\le n$ such that $\sup_{H\in T}\norm{H}_F\le K$, $\sup_{H\in T}\norm{H}_*\le K\sqrt{r}$, and $\sup_{H\in T}\abs{\Tr (H)}\le K$. Define
\[
\mathcal E_T:=\sup_{H\in T}\abs{\frac1m\sum_{i=1}^m\abs{\bm{a}_i^*H\bm{a}_i}-\E\abs{\bm{a}_i^*H\bm{a}_i}}.
\]
Then there exists a universal constant $C$ such that for every $u\ge 1$,
\begin{equation}\label{eq:sampletail}
\Pbb\left\{\mathcal E_T>CK\left(\sqrt{\frac{nr}{m}}+\frac{nr}{m}+\sqrt{\frac{u}{m}}+\frac{u}{m}\right)\right\}\le e^{-u}.
\end{equation}
\end{lemma}
\begin{lemma}\label{lem:meanbounds}
Let $\bm{a}\in\HH^n$ be a standard Gaussian vector and $2r\le n$. Then for any $X,Y\in \Mr(\HH)$, we have
\[
\frac{1}{\beta_0^{\HH,*}(r)}\norm{X-Y}_*\le \E \abs{\bm{a}^*(X-Y)\bm{a}} \le\norm{X-Y}_*,
\]
\[
(r+1)\eta_0^{\HH}(r)\norm{X-Y}_2\le \E \abs{\bm{a}^*(X-Y)\bm{a}} \le U^\HH_r \norm{X-Y}_2,
\]
where $U^{\HH}_r$ is defined in Theorem~\ref{thm:randombounds}.
\end{lemma}

The following lemma provides the estimates needed to verify the hypotheses of Lemma~\ref{lem:sampleerror}.
\begin{lemma}\label{lem:meanlower}
Let $\bm{a}\in\HH^n$ be a standard Gaussian vector. Then for any $H\in \Sn(\mathbb H)$, we have
\[
\E \abs{\bm{a}^*H\bm{a}}\ge \abs{\Tr(H)} \quad \text{and}\quad \E \abs{\bm{a}^*H\bm{a}}\gtrsim \norm{H}_F.
\]
\end{lemma}
The proof of Lemma~\ref{lem:meanlower} is given in the Appendix. We now prove Theorem~\ref{thm:randombounds}.

\begin{proof}[Proof of Theorem~\ref{thm:randombounds}]
Let
\[
T(r):=\left\{\frac{X-Y}{\E \abs{\bm{a}_i^*(X-Y)\bm{a}_i}}:  X, Y \in \Mr, X\neq Y \right\}.
\]

To apply Lemma~\ref{lem:sampleerror} with $T(r)$ in place of $T$, we first verify that its hypotheses are satisfied.
 By Lemma~\ref{lem:meanbounds} and the normalization in the definition of $T(r)$, we have
\[
\sup_{H\in T(r)}\norm{H}_*\le \beta_{0}^{\HH,*}(r)=O\left(\sqrt{r}\right),
\]
where the last equality follows from the definition of $\beta_{0}^{\HH,*}(r)$ in Definition~\ref{def:betaconstants} and the standard properties of the Gamma function. Lemma~\ref{lem:meanlower} gives $\sup_{H\in T(r)}\norm{H}_F=O(1)$ and $\sup_{H\in T(r)} \abs{\Tr(H)}\le 1$. Therefore the hypotheses of Lemma~\ref{lem:sampleerror} are satisfied with $r$ and a universal constant $K$.

Note that $\E\abs{\bm{a}_i^* H \bm{a}_i}=1$ holds for all $H\in T(r)$.
If $m\ge 16C^2K^2\delta^{-2}nr$, then applying Lemma~\ref{lem:sampleerror} to $T(r)$ with $u=\delta^2m/(16C^2K^2)$ shows that
\begin{equation}\label{eq:uniformerror}
\sup_{H\in T(r)}\abs{\frac1m\sum_{i=1}^m\abs{\bm{a}_i^*H\bm{a}_i}-1} \le \delta
\end{equation}
holds with probability at least $1-2\exp(-\delta^2 m/(16C^2K^2))$. Combining \eqref{eq:uniformerror} with Lemma~\ref{lem:meanbounds} yields, for all $X, Y \in \Mr(\HH)$,
\[
\frac{1}{\beta_{0}^{\HH,*}(r)}(1-\delta)\norm{X-Y}_* \le \frac{1}{m} \norm{\Phi_{\Aone}(X)-\Phi_{\Aone}(Y)}_1 \le (1+\delta)\norm{X-Y}_*,
\]
\[
(r+1)\eta_0^{\HH}(r)(1-\delta)\norm{X-Y}_2 \le \frac{1}{m} \norm{\Phi_{\Aone}(X)-\Phi_{\Aone}(Y)}_1 \le U^{\HH}_r(1+\delta)\norm{X-Y}_2 .
\]
Therefore we conclude that
\[
\beta_{\Aone}^{\HH,*}(r)
=
\frac{U_{\Aone}^{\HH,*}(r)}{L_{\Aone}^{\HH,*}(r)}
\leq
\beta_0^{\HH,*}(r)\frac{1+\delta}{1-\delta}\quad\text{and}\quad
\beta_{\Aone}^{\HH,2}(r)
=
\frac{U_{\Aone}^{\HH,2}(r)}{L_{\Aone}^{\HH,2}(r)}
\leq
\frac{U_r^{\HH}}{(r+1)\eta_0^{\HH}(r)}\frac{1+\delta}{1-\delta}.
\]
Moreover, Theorem~\ref{thm:lowerbound} gives $\beta_0^{\HH,*}(r)\le \beta_{\Aone}^{\HH,*}(r)$ and $\beta_0^{\HH,2}(r)\le \beta_{\Aone}^{\HH,2}(r)$. This completes the proof.
\end{proof}

\subsection{Proof of Lemma~\ref{lem:sampleerror} }

The proof of Lemma~\ref{lem:sampleerror} is based on a symmetrization and contraction argument, and the resulting supremum is controlled by the $(r,2)$-norm of an associated matrix. The following lemma gives the $L_p$ version of the symmetrization inequality.
\begin{lemma}\label{lem:signaverage}\textnormal{\cite[Lemma~2.3.1]{vanderVaartWellner1996}}
Let $X_1,\dots,X_m$ be i.i.d. random variables and $\eps_1,\dots,\eps_m$ be independent Rademacher signs, independent of $X_1,\dots,X_m$. Suppose $\mathcal F$ is a class of functions such that the left-hand side of \eqref{eq:signaverage} is finite. Then, for every $p\ge 1$,
\begin{equation}\label{eq:signaverage}
\norm{\sup_{f\in\mathcal F}\Big|\sum_{i=1}^m(f(X_i)-\E f(X_i))\Big|}_{L_p}
\le
2\norm{\sup_{f\in\mathcal F}\Big|\sum_{i=1}^m\eps_i f(X_i)\Big|}_{L_p}.
\end{equation}
\end{lemma}
The next lemma gives the $L_p$ version of the contraction inequality.
\begin{lemma}\label{lem:signbound}\textnormal{\cite[Theorem 4.12]{LedouxTalagrand1991}}
Let $T\subset\R^m$ and $\phi:\R\to\R$ be a $1$-Lipschitz function satisfying $\phi(0)=0$. Suppose $\eps_1,\dots,\eps_m$ are independent Rademacher signs. Then for each $p\ge 1$,
\[
\norm{\sup_{\bm{t}\in T}\Big|\sum_{i=1}^m\eps_i\phi(t_i)\Big|}_{L_p}
\le
2\norm{\sup_{\bm{t}\in T}\Big|\sum_{i=1}^m\eps_i t_i\Big|}_{L_p}.
\]
\end{lemma}
Lemma~\ref{lem:momenttail} converts the moment bound of a random variable to the tail bound. And Lemma~\ref{lem:tailmoment} is a converse statement.
\begin{lemma}\label{lem:momenttail}\textnormal{\cite[Lemma A.1]{Dirksen2015}}
Suppose $X$ is a random variable satisfying
\[
\norm{X}_{L_p}\le a_1+a_2\sqrt p+a_3p,
\]
for some $0< a_1,a_2,a_3<\infty$ and any $p\ge 1$.
Then for any $u\ge1$, we have
\[
\Pbb\left\{\abs{X}\ge e\left(a_1+a_2\sqrt u+a_3u\right)\right\}\le e^{-u}.
\]
\end{lemma}
\begin{lemma}\label{lem:tailmoment}\textnormal{\cite[Lemma A.2]{Dirksen2015}}
Suppose that $X$ is a random variable satisfying
\[
\Pbb\left\{\abs{X}\ge a_1\sqrt u+a_2u\right\}\le 2e^{-u},
\]
for some $0< a_1,a_2<\infty$ and any $u\ge 0$.
Then for all $p\ge 1$, we have
\[
\norm{X}_{L_p}\lesssim a_1\sqrt p+a_2p.
\]
\end{lemma}
To obtain the tail bound in Lemma~\ref{lem:sampleerror}, we use the following $(r,2)$-norm.
\begin{definition}\label{def:ranknorm}
For a Hermitian matrix $G\in\mathbb H^{n\times n}$, define the $(r,2)$-norm by
\[
\norm{G}_{(r,2)}:=\Big(\sum_{i=1}^r \sigma_i(G)^2\Big)^{1/2},
\]
where $\sigma_1(G)\ge \cdots \ge \sigma_n(G)$ denote the singular values of $G$.
\end{definition}

The following lemma provides an estimate for the $(r,2)$-norm used in our proof. The real-case result was established in \cite{GaoMaRenZhou2015}, and the complex case can follow from it. We prove this lemma in the Appendix.
\begin{lemma}\label{lem:matrixmoment}
Let $\bm{a}_1,\ldots,\bm{a}_m\in\mathbb H^n$ be independent standard Gaussian vectors.
Then for every integer $1\le r\le n$ and every $p\ge1$, we have
\begin{equation}\label{eq:matrixmoment}
\left(\E \Big\|\sum_{i=1}^m(\bm{a}_i \bm{a}_i^*-I_n)\Big\|_{(r,2)}^p\right)^{1/p}\lesssim \sqrt{mnr}+nr+\sqrt{mp}+p.
\end{equation}
\end{lemma}

We now have all the ingredients for proving Lemma~\ref{lem:sampleerror}.
\begin{proof}[Proof of Lemma~\ref{lem:sampleerror}]
Fix $p\ge1$. Applying Lemma~\ref{lem:signaverage} to $\{f_H(\bm{x}):=\abs{\bm{x}^* H\bm{x}} : H\in T\}$, and then applying Lemma~\ref{lem:signbound} with $\phi(x)=\abs{x}$, we obtain
\begin{equation}\label{eq:signmaximum}
\norm{\mathcal E_T}_{L_p}
\le \frac2m
\norm{\sup_{H\in T}\Big|\sum_{i=1}^m\eps_i\abs{\bm{a}_i^* H \bm{a}_i}\Big|}_{L_p} \le \frac 4m
\norm{\sup_{H\in T}\Big|\sum_{i=1}^m\eps_i \bm{a}_i^* H \bm{a}_i\Big|}_{L_p}.
\end{equation}
Denote $S_m:=\sum_{i=1}^m\eps_i$ and $Z_m:=\sum_{i=1}^m\eps_i(\bm{a}_i\bm{a}_i^*-I_n)$ .
Then \eqref{eq:signmaximum} becomes
\begin{equation}\label{eq:tracesplit}
\norm{\mathcal E_T}_{L_p} \le \frac4m \norm{\sup_{H\in T}\abs{\ip{Z_m}{H}+S_m\Tr H}}_{L_p} \le \frac4m \norm{\sup_{H\in T}\abs{\ip{Z_m}{H}}}_{L_p}+\frac4m \sup_{H\in T}\abs{\Tr H}\norm{S_m}_{L_p}.
\end{equation}
By the hypotheses of Lemma~\ref{lem:sampleerror}, $\norm{H}_F\le K$ and $\norm{H}_*\le K\sqrt{r}$ for some $K$ and integer $r$. Hence
\[
\begin{split}
\abs{\ip{Z_m}{H}}&\le \sum_{i=1}^n \sigma_i(Z_m)\sigma_i(H)\le \Big(\sum_{i=1}^r \sigma_i(Z_m)^2\Big)^{1/2}\Big(\sum_{i=1}^r \sigma_i(H)^2\Big)^{1/2}+ \sigma_{r+1}(Z_m) \sum_{i=r+1}^n \sigma_i(H) \\
&\le K\Big(\sum_{i=1}^r \sigma_i(Z_m)^2\Big)^{1/2}+K\sqrt{r}\sigma_{r+1}(Z_m)\le 2K\norm{Z_m}_{(r,2)},
\end{split}
\]
where the first inequality follows from von Neumann's trace inequality.
Therefore we have
\begin{equation}\label{eq:dualbound}
\sup_{\substack{\norm{H}_F\le K\\ \norm{H}_*\le K\sqrt{r}}}\abs{\ip{Z_m}{H}}\le2K\norm{Z_m}_{(r,2)}.
\end{equation}
We next bound $\norm{Z_m}_{(r,2)}$. Given the signs $\eps_i$, set $I_+(\bm{\eps}):=\{i:\eps_i=1\}$ and $I_-(\bm{\eps}):=\{i:\eps_i=-1\}$.
Then
\[
Z_m=\sum_{i\in I_+(\bm{\eps})}(\bm{a}_i\bm{a}_i^*-I_n)-\sum_{i\in I_-(\bm{\eps})}(\bm{a}_i\bm{a}_i^*-I_n).
\]
Therefore,
\[
\norm{Z_m}_{(r,2)}\le\bigg\|\sum_{i\in I_+(\bm{\eps})}(\bm{a}_i\bm{a}_i^*-I_n)\bigg\|_{(r,2)}+
\bigg\|\sum_{i\in I_-(\bm{\eps})}(\bm{a}_i\bm{a}_i^*-I_n)\bigg\|_{(r,2)}.
\]
Applying Lemma~\ref{lem:matrixmoment} for fixed signs $\eps_i$ and taking the $L_p$ norm over $\eps_i$ yields 
\begin{equation}\label{eq:signedmoment}
\begin{split}
\left(\E\norm{Z_m}_{(r,2)}^p\right)^{1/p}&\lesssim \norm{\sqrt{\abs{I_+(\bm{\eps})}nr}+nr+\sqrt{\abs{I_+(\bm{\eps})}p}+p}_{L_p} +\norm{\sqrt{\abs{I_-(\bm{\eps})}nr}+nr+\sqrt{\abs{I_-(\bm{\eps})}p}+p}_{L_p} \\
&\lesssim \sqrt{mnr}+nr+\sqrt{mp}+p,
\end{split}
\end{equation}
where the last inequality follows from $\abs{I_+(\bm{\eps})}\le m$ and $\abs{I_-(\bm{\eps})}\le m$. Also, Hoeffding's inequality gives
\[
\Pbb\{\abs{S_m}\ge t\} \le 2e^{-t^2/(2m)}.
\]
Lemma~\ref{lem:tailmoment} then yields
\begin{equation}\label{eq:signmoment}
\norm{S_m}_{L_p}\lesssim \sqrt{mp}.
\end{equation}
Substituting \eqref{eq:dualbound}, \eqref{eq:signedmoment} and \eqref{eq:signmoment} into \eqref{eq:tracesplit}, we obtain
\[
\norm{\mathcal E_T}_{L_p}\le\frac{CK}{m}\left(\sqrt{mnr}+nr+\sqrt{mp}+p\right)
=CK\left(\sqrt{\frac{nr}{m}}+\frac{nr}{m}+\sqrt{\frac{p}{m}}+\frac{p}{m}\right),
\]
where $C$ is a universal constant. The tail bound \eqref{eq:sampletail} follows from Lemma~\ref{lem:momenttail}.
\end{proof}

\subsection{Proof of Lemma~\ref{lem:meanbounds} }
To prove Lemma~\ref{lem:meanbounds}, we begin with some technical lemmas. The following lemma calculates expectations concerning two gamma-distributed random variables.

\begin{lemma}\label{lem:gammagap}
Let $s,\theta$ be positive constants. Suppose that $G_1$ and $G_2$ are independent Gamma$(s,\theta)$ random variables. Then
\[
\E\abs{G_1-G_2}=2\theta\frac{\Gamma\left((2s+1)/2\right)}{\sqrt{\pi}\Gamma(s)}.
\]
Moreover, if $G_1\sim\chi^2(2)$ and $G_2\sim\chi^2(1)$ are independent, then $\E\abs{G_1-G_2}=2\sqrt{2}-1$.
\end{lemma}

The proof of Lemma~\ref{lem:gammagap} is given in the Appendix. We determine the lower endpoints of the expectations appearing in Lemma~\ref{lem:meanbounds} in the next lemma.
\begin{lemma}\label{lem:meanminimum}
Let $Y_1,\ldots,Y_{2r}$ be i.i.d. positive random variables with $\E Y_i<\infty$. Define the
function
\[
f(\bm{x}):=\E\abs{\sum_{i=1}^{2r}x_iY_i}.
\]
Set $X_r=\{\bm{x}\in\R^{2r}: \#\{i:x_i>0\}\le r, \#\{i:x_i<0\}\le r\}$. Then $f$ is convex and
\begin{equation}\label{eq:summinimum}
\inf_{\bm{x}\in X_r, \norm{\bm{x}}_1=1} f(\bm{x})=f\Bigl(\underbrace{\frac{1}{2r},\dots,\frac{1}{2r}}_{r\text{ times}},\underbrace{-\frac{1}{2r}, \dots, -\frac{1}{2r}}_{r\text{ times}}\Bigr),
\end{equation}
\begin{equation}\label{eq:peakminimum}
\inf_{\bm{x}\in X_r, \norm{\bm{x}}_{\infty}=1} f(\bm{x})=\inf_{t\in[0,1]} f(1,\underbrace{-t,\dots,-t}_{r\text{ times}}, \underbrace{0, \dots, 0}_{r-1\text{ times}}).
\end{equation}
\end{lemma}

To prove Lemma~\ref{lem:meanminimum}, we introduce Lemma~\ref{lem:convexorder} to establish the relevant convexity property.
\begin{lemma}\label{lem:convexorder}\textnormal{\cite[Proposition~4.C.1]{MOA2011}}
Let $\bm{x}, \bm{y} \in \R^n$ and write their components in decreasing order as $x_{(1)}\ge x_{(2)}\ge \dots \ge x_{(n)}$ and $y_{(1)}\ge y_{(2)}\ge \dots \ge y_{(n)}$, respectively. We define the majorization relation as $\bm{y} \prec \bm{x}$ if and only if
\[
\sum_{i=1}^k y_{(i)} \le \sum_{i=1}^k x_{(i)} \text{ for } k=1,\dots, n-1, \quad \text{and}\quad
\sum_{i=1}^n y_{(i)} = \sum_{i=1}^n x_{(i)}.
\]
Then $\bm{y}\prec \bm{x}$ if and only if
\[
\bm{y} \in \operatorname{conv}\{P_\sigma \bm{x}: P_\sigma \text{ is a permutation matrix }\},
\]
where $\operatorname{conv} X$ denotes the convex hull of $X$.
\end{lemma}

\begin{proof}[Proof of Lemma~\ref{lem:meanminimum}]
For every $\bm{x}, \bm{x}',\bm{y}\in \R^{2r}$ and $\lambda\in [0,1]$, the convexity of the absolute value function gives
\begin{equation}\label{eq:convexbound}
\abs{\ip{\lambda \bm{x}+(1-\lambda)\bm{x}'}{\bm{y}}}
\le
\lambda \abs{\ip{\bm{x}}{\bm{y}}}+(1-\lambda)\abs{\ip{\bm{x}'}{\bm{y}}}.
\end{equation}
Applying \eqref{eq:convexbound} pointwise and then taking expectations shows that $f$ is convex.

We first prove \eqref{eq:summinimum}. For $t\in \R$, define
\[
\bm{a}_t:=
\Bigl(\underbrace{t,\dots,t}_{r\text{ times}},
\underbrace{t-\frac{1}{r},\dots,t-\frac{1}{r}}_{r\text{ times}}\Bigr).
\]
Given $\bm{x}\in X_r$ with $\norm{\bm{x}}_1=1$, we may permute its coordinates so that $x_1\ge x_2\ge \cdots \ge x_{2r}$. Since $\bm{x}$ has at most $r$ positive and at most $r$ negative coordinates, this ordering also gives $x_r\ge 0$ and $x_{r+1}\le 0$. Set the average of coordinates $t(\bm{x}):=\sum_{i=1}^{r}x_i /r$. Then $\sum_{i=1}^k x_i \ge kt(\bm{x})$ and $\sum_{i=r+1}^{r+k} x_i \ge k(t(\bm{x})-1/r)$ for all $1\le k \le r$. Moreover, $\sum_{i=1}^{2r} x_i=2rt(\bm{x})-1$. Hence, $\bm{a}_{t(\bm{x})} \prec \bm{x}$ in the sense of Lemma~\ref{lem:convexorder}. Therefore $\bm{a}_{t(\bm{x})}$ belongs to the convex hull of the permutations of $\bm{x}$ by Lemma~\ref{lem:convexorder}. Since
$f$ is convex and permutation invariant, we obtain
\begin{equation}\label{eq:sumflatten}
f(\bm{x})\ge f\left(\bm{a}_{t(\bm{x})}\right).
\end{equation}
Thus the infimum of $f$ is attained by some vector of the form $\bm{a}_t$ with $t\in [0,1/r]$. And the problem reduces to the one-parameter family $\bm{a}_t$. Viewed as a function of $t$, $f(\bm{a}_t)$ is convex and symmetric about $1/(2r)$. Hence the minimum of $f(\bm{a}_t)$ is attained at $1/(2r)$, which proves \eqref{eq:summinimum}.

Next we prove \eqref{eq:peakminimum}. For $t\in \R$, define
\[
\bm{b}_t:=
(1,\underbrace{-t,\dots,-t}_{r\text{ times}},
\underbrace{0,\dots,0}_{r-1\text{ times}}).
\]
Given $\bm{x}\in X_r$ with $\norm{\bm{x}}_{\infty}=1$, we may replace $\bm{x}$ by $-\bm{x}$ if necessary, and then permute its coordinates so that $x_1=1$, and $x_1\ge x_2\ge \cdots \ge x_{2r}$. Set $t(\bm{x}):=-\sum_{i=2}^{2r}x_i /r \in [-1,1]$. Then $\sum_{i=2r-k}^{2r} x_i\le -(k+1) t(\bm{x})$ for every $0\le k< r$, and $\sum_{i=k}^{2r} x_i\le -r t(\bm{x})$ for every $2\le k\le r$. Moreover, $\sum_{i=1}^{2r} x_i=1-rt(\bm{x})$. Hence, $\bm{b}_{t(\bm{x})} \prec \bm{x}$, and $\bm{b}_{t(\bm{x})}$ belongs to the convex hull of the permutations of $\bm{x}$ by Lemma~\ref{lem:convexorder}. Thus
\begin{equation}\label{eq:peakflatten}
f(\bm{x})\ge f\left(\bm{b}_{t(\bm{x})}\right).
\end{equation}
Since each $Y_i$ is a positive random variable, we have $f(1,0,\dots,0)\le f(\bm{b}_t)$ for $t\le 0$. Thus the infimum of $f$ is attained by some vector of the form $\bm{b}_t$ with $t\in [0,1]$, which proves \eqref{eq:peakminimum}.
\end{proof}

Now we have all the ingredients needed to prove Lemma~\ref{lem:meanbounds}.
\begin{proof}[Proof of Lemma~\ref{lem:meanbounds}]
Write the spectral decomposition of $X-Y$ as $X-Y=U^*DU$, where $D$ is a real diagonal matrix and $U$ is unitary. Since $\bm{a}$ is a standard Gaussian vector, $U\bm{a}$ is also a standard Gaussian vector. Writing $D=\diag(\lambda_1,\dots,\lambda_n)$, we obtain
\[
\bm{a}^*(X-Y)\bm{a}\overset{d}{=}\sum_{i=1}^n \lambda_i \abs{a_i}^2.
\]
Note that $\abs{a_i}^2\sim \chi^2(1)$ if $\HH=\R$ and $\abs{a_i}^2\sim \Gamma(1,1)$ if $\HH=\C$. Since $X$ and $Y$ are of rank at most $r$, we may assume $\lambda_{2r+1}=\dots=\lambda_n=0$. Since $X$ and $Y$ are positive semidefinite, at most $r$ of the $\lambda_i$ are positive and at most $r$ are negative. Therefore it suffices to prove that for any $\bm{x}\in\R^{2r}$ satisfying $\#\{i:x_i>0\}\le r$ and $\#\{i:x_i<0\}\le r$, we have
\begin{equation}\label{eq:sumbounds}
\frac{1}{\beta_0^{\HH,*}(r)}\norm{\bm{x}}_1\le \E\abs{\sum_{i=1}^{2r}x_i\abs{a_i}^2} \le\norm{\bm{x}}_1,
\end{equation}
\begin{equation}\label{eq:peakbounds}
(r+1)\eta_0^{\HH}(r)\norm{\bm{x}}_{\infty}\le \E\abs{\sum_{i=1}^{2r}x_i\abs{a_i}^2}\le U^\HH_r \norm{\bm{x}}_\infty.
\end{equation}

We first prove \eqref{eq:sumbounds}. Assume $\norm{\bm{x}}_1=1$. For the upper bound, the triangle inequality immediately gives
\begin{equation}\label{eq:sumupper}
\E\abs{\sum_{i=1}^{2r} x_i \abs{a_i}^2}
\le
\sum_{i=1}^{2r}\abs{x_i}\E \abs{a_i}^2=1.
\end{equation}
For the lower bound, $\sum_{i=1}^r \abs{a_i}^2\sim \Gamma(r/2,2)$ if $\HH=\R$, and $\sum_{i=1}^r \abs{a_i}^2\sim \Gamma(r,1)$ if $\HH=\C$. Using Lemma~\ref{lem:meanminimum} and Lemma~\ref{lem:gammagap}, we obtain
\[
\E\abs{\sum_{i=1}^{2r}x_i\abs{a_i}^2}\ge \E\abs{\sum_{i=1}^{r}\frac{1}{2r} \abs{a_i}^2-\sum_{i=r+1}^{2r}\frac{1}{2r} \abs{a_i}^2}=\frac{1}{\beta_0^{\HH,*}(r)}.
\]

We next prove \eqref{eq:peakbounds}. Assume $\norm{\bm{x}}_\infty=1$. For the lower bound, Lemma~\ref{lem:gammaratio} implies that $\sum_{i=1}^{r+1} \abs{a_i}^2$ and $\abs{a_1}^2/\sum_{i=1}^{r+1} \abs{a_i}^2$ are independent. Moreover,
\[
\frac{\abs{a_1}^2}{\sum_{i=1}^{r+1} \abs{a_i}^2} \sim \begin{cases}
\mathrm{Beta}\left(\frac12,\frac r2\right), & \HH=\R,\\
\mathrm{Beta}(1,r), & \HH=\C.
\end{cases}
\]
Therefore, we obtain
\begin{equation}\label{eq:reducedmean}
\E\abs{\abs{a_1}^2-t\sum_{i=2}^{r+1} \abs{a_i}^2}
=
\E \sum_{i=1}^{r+1} \abs{a_i}^2\cdot \E\abs{(1+t)\frac{\abs{a_1}^2}{\sum_{i=1}^{r+1} \abs{a_i}^2}-t}=(r+1)\eta^{\HH}(r,t).
\end{equation}
Applying Lemma~\ref{lem:meanminimum} and then minimizing over $t$, we conclude that
\[
\E\abs{\sum_{i=1}^{2r}x_i\abs{a_i}^2}\ge \inf_{t\in[0,1]} \E\abs{\abs{a_1}^2-t\sum_{i=2}^{r+1} \abs{a_i}^2}= (r+1)\eta_0^{\HH}(r).
\]

It remains to establish the upper bound in \eqref{eq:peakbounds}.
Fix a sign pattern for $\bm{x}$. Under the normalization $\norm{\bm{x}}_{\infty}=1$, the feasible set for each coordinate is either $[0,1]$ or $[-1,0]$. Note that the map $\bm{x}\to \E\abs{\sum_{i=1}^{2r}x_i\abs{a_i}^2}$ is convex. The maximum of a convex function over its domain is attained at an extreme point. Hence, it is enough to consider vectors $\bm{x}$ with each coordinate $x_i\in\{0,-1,1\}$. For such $\bm{x}$, let $p$ and $q$ denote the numbers of coordinates equal to $+1$ and $-1$, respectively. By symmetry, we may assume $p\ge q$. If $q=0$, then $\E\abs{\sum_{i=1}^{p}\abs{a_i}^2}=p\le r$. Now assume $q\ge 1$. In the complex case, applying Lemma~\ref{lem:gammagap} to $\Gamma(q,1)$ random variables gives
\[
\E \abs{\sum_{i=1}^p \abs{a_i}^2-\sum_{i=p+1}^{p+q} \abs{a_i}^2} \le \E \sum_{i=1}^{p-q} \abs{a_i}^2+\E \abs{\sum_{i=p-q+1}^p \abs{a_i}^2-\sum_{i=p+1}^{p+q} \abs{a_i}^2}\le p-q+\frac{2\Gamma(q+1/2)}{\sqrt{\pi}\Gamma(q)}\le p \le r.
\]
In the real case, if $q\ge 2$, applying Lemma~\ref{lem:gammagap} to $\Gamma(q/2,2)$ random variables yields
\[
\E \abs{\sum_{i=1}^p a_i^2-\sum_{i=p+1}^{p+q} a_i^2} \le \E \sum_{i=1}^{p-q} a_i^2+\E \abs{\sum_{i=p-q+1}^p a_i^2-\sum_{i=p+1}^{p+q} a_i^2}\le p-q+\frac{4\Gamma((q+1)/2)}{\sqrt{\pi}\Gamma(q/2)} \le p\le r.
\]
It remains to treat the case $q=1$. If $p=1$, then applying Lemma~\ref{lem:gammagap} to $\Gamma(1/2,2)$ random variables yields
\[
\E\abs{a_1^2-a_2^2}=\frac{4}{\pi}.
\]
If $p\ge 2$, then Lemma~\ref{lem:gammagap} gives
\[
\E \abs{\sum_{i=1}^p a_i^2-a_{p+1}^2} \le \E \sum_{i=1}^{p-2} a_i^2 +\E \abs{a_{p-1}^2+a_p^2-a_{p+1}^2}=p-2+2\sqrt{2}-1\le p \le r.
\]
This completes the proof.
\end{proof}

\section{Stability of the PhaseLift Model}
We prove Lemma~\ref{lem:meanvalues} and Theorem~\ref{thm:robustrecovery} in this section.
\begin{proof}[Proof of Lemma~\ref{lem:meanvalues}]
We may assume $n\ge 2$. Given $H\in\Sn(\HH)$, write $H=U^*\diag(\lambda_1,\ldots,\lambda_n)U$, where $U$ is a unitary matrix. Then $\bm a^*H\bm a \overset{d}{=} \sum_{i=1}^n \lambda_i \abs{a_i}^2$. For $\bm{x}\in\R^{n}$, define
\[
f_n^{\HH}(\bm{x}):=\E \abs{\sum_{i=1}^n x_i \abs{a_i}^2}.
\]
Then
\[
\mu^{\HH,*}(r)=\inf_{n\ge r}\inf_{\substack{\norm{\bm{x}}_1=1, \\ \#\{i:x_i>0\}\le r}} f_n^{\HH}(\bm{x}),\qquad
\mu^{\HH,2}(r)=\inf_{n\ge r}\inf_{\substack{\norm{\bm{x}}_{\infty}=1, \\ \#\{i:x_i>0\}\le r}} f_n^{\HH}(\bm{x}).
\]

First we prove the assertion about the median. For $t\in [-n+1,n-1]$, define $\bm{a}_{t,n}=(1, -t/(n-1), \dots ,-t/(n-1))$.
Given $\bm{x}$ with $\norm{\bm{x}}_{\infty}=1$, permute its coordinates such that $\abs{x_1}=1$.

If $x_1=1$, set $t(\bm x)=-\sum_{i=2}^n x_i$. Then $\bm{a}_{t(\bm{x}),n} \prec \bm{x}$ in the sense of Lemma~\ref{lem:convexorder} and $\bm{a}_{t(\bm{x}),n}$ belongs to the convex hull of the permutations of $\bm{x}$. Since $f_n^{\HH}$ is convex and permutation invariant, $f_n^{\HH}(\bm{x})\ge f_n^{\HH}(\bm{a}_{t(\bm{x}),n})$.

If $x_1=-1$, then setting $t(\bm x)=\sum_{i=2}^n x_i$ yields that $\bm{a}_{t(\bm{x}),n} \prec \bm{-x}$ in the sense of Lemma~\ref{lem:convexorder} and hence $f_n^{\HH}(\bm{-x})\ge f_n^{\HH}(\bm{a}_{t(\bm{x}),n})$. Note that $f_n^{\HH}(\bm{-x})=f_n^{\HH}(\bm{x})$. Thus in either case we have $f_n^{\HH}(\bm{x})\ge f_n^{\HH}(\bm{a}_{t(\bm{x}),n})$.

For $t\in [0,n-1]$, $\norm{\bm{a}_{t,n}}_{\infty}=1$ and $\bm{a}_{t,n}$ has one positive entry. Therefore for fixed $n$,
\begin{equation}\label{eq:peakreduce}
\inf_{\substack{\norm{\bm{x}}_{\infty}=1, \\ \#\{i:x_i>0\}\le r}} f_n^{\HH}(\bm{x})\ge \inf_{t\in [-n+1,n-1]} f_n^{\HH}(\bm{a}_{t,n})= \inf_{t\in [0,n-1]} f_n^{\HH}(\bm{a}_{t,n})\ge \inf_{\substack{\norm{\bm{x}}_{\infty}=1, \\ \#\{i:x_i>0\}\le r}} f_n^{\HH}(\bm{x}),
\end{equation}
where the equality follows from $f_n^{\HH}(\bm{a}_{t,n})>f_n^{\HH}(\bm{a}_{0,n})$ for $t<0$.

For fixed $t$, set $\bm{a_{t,n}^0}:=(1, 0, -t/(n-1), \allowbreak\dots ,\allowbreak -t/(n-1))\in\R^{n+1}$. Then $\bm{a_{t,n+1}} \prec \bm{a_{t,n}^0}$. Thus
\begin{equation}\label{eq:dimensionstep}
f_n^{\HH}(\bm{a}_{t,n})=f_{n+1}^{\HH}\left(\bm{a}_{t,n}^0\right)\ge f_{n+1}^{\HH}(\bm{a}_{t,n+1}).
\end{equation}
The law of large numbers in $L^1$ gives
\begin{equation}\label{eq:meanlimit}
\lim_{n\to \infty} f_n^{\HH}(\bm{a}_{t,n})=\lim_{n\to \infty}\E \abs{\abs{a_1}^2-\frac{t}{n-1}\sum_{i=2}^n \abs{a_i}^2}=\E \abs{\abs{a_1}^2-t}.
\end{equation}
Combining \eqref{eq:peakreduce}, \eqref{eq:dimensionstep} and \eqref{eq:meanlimit}, we obtain
\[
\mu^{\HH,2}(r)=\inf_{n> r} \inf_{t\in [0,n-1]} f_n^{\HH}(\bm{a}_{t,n})=\inf_{t\ge 0}\lim_{n\to \infty} f_n^{\HH}(\bm{a}_{t,n})=\inf_{t\ge 0} \E \abs{\abs{a_1}^2-t}=\E \abs{\abs{a_1}^2-m_{\HH}},
\]
where $m_{\HH}$ is the median of $\abs{a_1}^2$. Since $\abs{a_1}^2 \sim \chi^2(1)$ if $\HH=\R$, and $\abs{a_1}^2 \sim \Gamma(1,1)$ if $\HH=\C$, this proves the assertion about the median. Note that \eqref{eq:reducedmean} gives $f_{n+1}^{\HH}(\bm{a}_{t,n+1})=(n+1)\eta^{\HH}(n,t/n)$. Taking the infimum yields
\[
\lim_{r \to \infty} (r+1)\eta_0^{\HH}(r)=\lim_{r \to \infty}\inf_t f_{r+1}^{\HH}(\bm{a}_{t,r+1})=\mu^{\HH,2}.
\]

Now we turn to $\mu^{\HH,*}(r)$, using an analogous argument. For $t\in [0,1]$, define
\[
\bm{b}_{t,n,k}:=
\Bigl(
\underbrace{\frac tk,\dots,\frac tk}_{k\text{ times}},
\underbrace{-\frac{1-t}{n-k},\dots,-\frac{1-t}{n-k}}_{n-k\text{ times}}\Bigr).
\]
Then $\norm{\bm{b}_{t,n,k}}_1=1$. Given $\bm{x}$ with $\norm{\bm{x}}_1=1$, permute its coordinates such that $x_1\ge x_2\ge\dots\ge x_n$. If $x_1\le 0$ or $x_n\ge 0$, then $f_n^{\HH}(\bm{x})=1$. Assume that $x_k> 0$ and $x_{k+1}\le 0$ for some $1\le k\le n-1$. Set $t(\bm x)=\sum_{i=1}^k x_i$. Then $\bm{b}_{t(\bm{x}),n,k} \prec \bm{x}$ in the sense of Lemma~\ref{lem:convexorder} and $\bm{b}_{t(\bm{x}),n,k}$ belongs to the convex hull of the permutations of $\bm{x}$. Thus
\[
f_n^{\HH}(\bm{x})\ge f_n^{\HH}\left(\bm{b}_{t(\bm{x}),n,k}\right).
\]
Hence, for fixed $n$, the infimum is attained at a vector of the form $\bm{b}_{t,n,k}$. For fixed $t$, define
\[
\bm{b}_{t,n,k}^0:=
\Bigl(
\underbrace{\frac tk,\dots,\frac tk}_{k\text{ times}},
\underbrace{0,\dots,0}_{r+1-k\text{ times}},
\underbrace{-\frac{1-t}{n-k},\dots,-\frac{1-t}{n-k}}_{n-k\text{ times}}\Bigr).
\]
Then $\bm{b}_{t,n+1+r-k,r} \prec \bm{b}_{t,n,k}^0$. Thus $f_n^{\HH}(\bm{b}_{t,n,k})=f_{n+1+r-k}^{\HH}(\bm{b}_{t,n,k}^0)\ge f_{n+1+r-k}^{\HH}(\bm{b}_{t,n+1+r-k,r})$. Let $Y_r\sim\chi^2(r)$ when $\HH=\R$, and let $Y_r\sim\Gamma(r,1)$ when $\HH=\C$. Then
\[
\mu^{\HH,*}(r)=\inf_{n> r} \inf_{k\le r} \inf_t f_n^{\HH}(\bm{b}_{t,n,k})=\inf_t\lim_{n\to \infty} f_n^{\HH}(\bm{b}_{t,n,r})=\inf_t \E \abs{\frac{t}{r}Y_r-(1-t)},
\]
where the last equality follows from the law of large numbers in $L^1$ .

We now determine the limit of $\mu^{\HH,*}(r)\beta_0^{\HH,*}(r)$. Stirling's formula gives $\beta_0^{\R,*}(r)\sim\sqrt{\pi r/2}$ and $\beta_0^{\C,*}(r)\sim\sqrt{\pi r}$. Let $\psi_r^{\HH}(t):=\E|(t/r) Y_r-(1-t)|$. The central limit theorem together with uniform integrability gives
\[
\lim_{r\to \infty} \E\abs{\sqrt{r}\left(\frac 1r Y_r-1\right)} =\lim_{r\to \infty}\E\abs{\sum_{i=1}^r \frac{1}{\sqrt{r}}\left(\abs{a_i}^2-\E \abs{a_i}^2\right)}=\E\abs{\sigma_{\HH}\cdot G},
\]
where $G\sim N(0,1)$, $\sigma_{\HH}=\sqrt2$ if $\HH=\R$, and $\sigma_{\HH}=1$ if $\HH=\C$. Since $\E \abs{G}=\sqrt{2/\pi}$, we have
\[
\limsup_{r \to \infty} \mu^{\HH,*}(r)\beta_0^{\HH,*}(r) \le \lim_{r \to \infty} \psi_r^{\HH}\left(\frac{1}{2}\right)\beta_0^{\HH,*}(r) = \frac{1}{\sqrt{2}}.
\]
Suppose $t_r$ minimizes $\psi_r^{\HH}$. Then Jensen's inequality gives
\[
\mu^{\HH,*}(r)=\psi_r^{\HH}(t_r)\ge \abs{\frac{t_r}{r}\E Y_r-(1-t_r)}=\abs{2t_r-1}.
\]
Combining this inequality with the upper bound above shows that $\sqrt r(2t_r-1)$ is bounded. Fix any subsequence of $t_r$. By passing to a further subsequence if necessary, we may assume that $\lim_{r\to\infty} \sqrt r(2t_r-1)= c$. Then $\lim_{r\to \infty}t_r= 1/2$ and
\[
\lim_{r\to \infty}\sqrt r\mu^{\HH,*}(r)
=
\lim_{r\to \infty}\E\abs{\sqrt{r}(2t_r-1)+t_r\frac{Y_r-r}{\sqrt r}}
=\E\abs{c+\frac{\sigma_{\HH}}{2}G}.
\]
Since $G$ is symmetric, the function $c\mapsto \E|c+\sigma_{\HH}G/2|$ is minimized at $c=0$. Hence every subsequential limit of $\sqrt r\mu^{\HH,*}(r)$ is at
least $\E\abs{\sigma_{\HH}G/2}=\sigma_{\HH}\sqrt{1/(2\pi)}$.
Consequently,
\[
\liminf_{r\to\infty}\sqrt r\mu^{\R,*}(r)\ge \frac{1}{\sqrt{\pi}},
\qquad
\liminf_{r\to\infty}\sqrt r\mu^{\C,*}(r)\ge \frac{1}{\sqrt{2\pi}}.
\]
Thus
\[
\lim_{r \to \infty} \mu^{\HH,*}(r)\beta_0^{\HH,*}(r) = \frac{1}{\sqrt{2}}.
\]
This completes the proof.
\end{proof}

We next prove Theorem~\ref{thm:robustrecovery} by applying Lemma~\ref{lem:sampleerror} analogously to the proof of Theorem~\ref{thm:randombounds}.
\begin{proof}[Proof of Theorem~\ref{thm:robustrecovery}]
Let $S\subset [m]$ denote the support of $\bm{\eta}$. For fixed $\hat{X}_{1,+}$, set $D:=X_0-\hat{X}_{1,+}$. The optimality of $\hat{X}_{1,+}$ gives $\norm{\bm{\eta}+\bm{w}}_1\ge \norm{\bm{\eta}+\bm{w}+\Phi_{\Aone}(D)}_1$. Hence
\begin{equation}\label{eq:residualbound}
\norm{\bm{\eta}+\Phi_{\Aone}(D)}_1\le \norm{\bm{\eta}+\bm{w}}_1+\norm{\bm{w}}_1\le \norm{\bm{\eta}}_1+2\norm{\bm{w}}_1.
\end{equation}
Since $\abs{a+b}-\abs{a}\ge \sgn(a)b$ for all $a,b\in\R$, \eqref{eq:residualbound} gives
\begin{equation}\label{eq:noisebound}
\sum_{i\notin S} \abs{\bm{a}_i^*D\bm{a}_i}+\sum_{i\in S} \sgn(\eta_i) \bm{a}_i^*D\bm{a}_i \le \norm{\bm{\eta}+\Phi_{\Aone}(D)}_1-\norm{\bm{\eta}}_1\le 2\norm{\bm{w}}_1.
\end{equation}

Next we bound $\sum_{i\notin S} \abs{\bm{a}_i^*D\bm{a}_i}+\sum_{i\in S} \sgn(\eta_i) \bm{a}_i^*D\bm{a}_i$ from below in terms of $\norm{D}_*$ and $\norm{D}_2$. Since both $\hat{X}_{1,+}$ and $X_0$ are positive semidefinite matrices and $X_0$ is of rank at most $r$, $D=X_0-\hat{X}_{1,+}$ has at most $r$ positive eigenvalues. Set
\[
B(r):=\left\{\frac{H}{\E \abs{\bm{a}_i^*H\bm{a}_i}}:H\in T^2(n,r)\right\}.
\]
Note that $\frac{D}{\E \abs{\bm{a}_i^*D\bm{a}_i}}\in B(r) $.
Replacing $T^2(n,r)$ by $T^*(n,r)$ in this definition does not change $B(r)$. Note that Lemma~\ref{lem:meanvalues} implies that $\mu^{\HH,*}(r)$ is of the same order as  $1/\beta_0^{\HH,*}(r)$. Since $\beta_0^{\HH,*}(r)\asymp r^{1/2}$, we have
\[
\sup_{H\in B(r)}\norm{H}_*\le 1/\mu^{\HH,*}(r)=O(\sqrt{r}).
\]
 Lemma~\ref{lem:meanlower} gives $\sup_{H\in B(r)}\norm{H}_F=O(1)$ and $\sup_{H\in B(r)}\abs{\Tr(H)}\le1$. Therefore the hypotheses of Lemma~\ref{lem:sampleerror} are satisfied with $r$ and a universal constant $K$. Also, $\E\abs{\bm{a}_i^* H \bm{a}_i}=1$ holds for all $H\in B(r)$.
Applying Lemma~\ref{lem:sampleerror} to $B(r)$ with $u=\delta^2(1-2s)^2m/(64C^2K^2)$ yields
\begin{equation}\label{eq:spliterror}
\begin{split}
&\sup_{H\in B(r)}  \abs{\frac 1m\sum_{i\notin S} \abs{\bm{a}_i^*H\bm{a}_i}-(1-s)}+ \sup_{H\in B(r)} \abs{\frac 1m \sum_{i\in S} \abs{\bm{a}_i^*H\bm{a}_i}-s} \\
\le&
CK\left(\sqrt{\frac{(1-s)nr}{m}}+\frac{nr}{m}+\sqrt{\frac{(1-s)u}{m}}+\frac{u}{m}+\sqrt{\frac{snr}{m}}+\frac{nr}{m}+\sqrt{\frac{su}{m}}+\frac{u}{m}\right) \\
\le&
2CK\left(\sqrt{\frac{nr}{m}}+\frac{nr}{m}+\sqrt{\frac{u}{m}}+\frac{u}{m}\right)\le (1-2s)\delta
\end{split}
\end{equation}
holds with probability at least $1-4\exp(-\delta^2(1-2s)^2 m/(64C^2K^2))$, provided that $m\ge 64C^2K^2\delta^{-2}(1-2s)^{-2}nr$. On the event that \eqref{eq:spliterror} holds, we have
\begin{equation}\label{eq:cleanbound}
\begin{split}
&\inf_{H\in B(r)} \left(\frac 1m \sum_{i\notin S} \abs{\bm{a}_i^*H\bm{a}_i}+ \frac 1m \sum_{i\in S} \sgn(\eta_i) \bm{a}_i^*H\bm{a}_i\right)
\ge
\inf_{H\in B(r)}\left(\frac 1m \sum_{i\notin S} \abs{\bm{a}_i^*H\bm{a}_i}-\frac 1m \sum_{i\in S} \abs{ \bm{a}_i^*H\bm{a}_i}\right) \\
\ge &
(1-s)- \sup_{H\in B(r)} \abs{\frac 1m\sum_{i\notin S} \abs{\bm{a}_i^*H\bm{a}_i}-(1-s)}-s-\sup_{H\in B(r)}  \abs{\frac 1m\sum_{i\in S} \abs{\bm{a}_i^*H\bm{a}_i}-s}
\ge (1-\delta)(1-2s).
\end{split}
\end{equation}

By Definition~\ref{def:meanminimum}, we have
$\mu^{\HH,2}(r)\norm{D}_2 \le \E\abs{\bm{a}_i^*D\bm{a}_i}$
and
$\mu^{\HH,*}(r)\norm{D}_* \le \E\abs{\bm{a}_i^*D\bm{a}_i}$.
Hence, \eqref{eq:noisebound} gives 
\[
\frac{1}{\E\abs{\bm{a}_i^*D\bm{a}_i}} \sum_{i\notin S} \abs{\bm{a}_i^*D\bm{a}_i}+ \frac{1}{\E\abs{\bm{a}_i^*D\bm{a}_i}}\sum_{i\in S} \sgn(\eta_i) \bm{a}_i^*D\bm{a}_i \le \frac{2}{\mu^{\HH,2}(r)\norm{D}_2} \norm{\bm{w}}_1,
\]
\[
\frac{1}{\E\abs{\bm{a}_i^*D\bm{a}_i}} \sum_{i\notin S} \abs{\bm{a}_i^*D\bm{a}_i}+ \frac{1}{\E\abs{\bm{a}_i^*D\bm{a}_i}}\sum_{i\in S} \sgn(\eta_i) \bm{a}_i^*D\bm{a}_i \le \frac{2}{\mu^{\HH,*}(r)\norm{D}_*} \norm{\bm{w}}_1,
\]
Since $D/\E\abs{\bm{a}_i^*D\bm{a}_i}\in B(r)$, \eqref{eq:cleanbound} yields 
\[
(1-\delta)(1-2s)\mu^{\HH,2}(r)\norm{X_0-\hat{X}_{1,+}}_2
\le
\frac 2m \norm{\bm{w}}_1,
\]
\[
(1-\delta)(1-2s)\mu^{\HH,*}(r)\norm{X_0-\hat{X}_{1,+}}_*
\le
\frac 2m \norm{\bm{w}}_1.
\]
Provided that $m\ge 64C^2K^2\delta^{-2}(1-2s)^{-2}nr$, both bounds hold simultaneously for every $X_0\in \Mr(\HH)$ with probability at least $1-4\exp(-1/(64C^2K^2)\delta^2(1-2s)^2 m)$. This completes the proof.
\end{proof}

\section{APPENDIX}

\subsection{Proof of Proposition~\ref{prop:rankone}}

The exceptional real rank-one case requires a separate two-dimensional argument. We first reduce the problem to a two-dimensional subspace and then compare the upper and lower Lipschitz constants by averaging over the unit circle.

\begin{proof}[Proof of Proposition~\ref{prop:rankone}]
For each $i$, denote the submatrix consisting of the first two rows and first two columns of $A_i$ by $B_i \in \mathbb{S}_2^+(\R)$, and set $\B:=\{B_i\}_{i=1}^m$.
By the definition of $L_\A^{\R,2}(1)$ and $U_{\A}^{\R,2}(1)$ we have
\[
L_{\A}^{\R,2}(1)=\inf_{\substack{X,Y\in \Snp(\R)\\ \rank(X)\le 1 \\ \rank(Y)\le 1}}\frac{\norm{\PhiA(X)-\PhiA(Y)}_1}{\norm{X-Y}_2}\le \inf_{\substack{X,Y\in \mathbb{S}_2^+(\R)\\ \rank(X)\le 1 \\ \rank(Y)\le 1}}\frac{\norm{\PhiB(X)-\PhiB(Y)}_1}{\norm{X-Y}_2}=L_{\B}^{\R,2}(1),
\]
\[
U_{\A}^{\R,2}(1):=\sup_{\substack{X,Y\in \Snp(\R)\\ \rank(X)\le 1 \\ \rank(Y)\le 1}}\frac{\norm{\PhiA(X)-\PhiA(Y)}_1}{\norm{X-Y}_2}\ge \sup_{\substack{X,Y\in \mathbb{S}_2^+(\R)\\ \rank(X)\le 1 \\ \rank(Y)\le 1}}\frac{\norm{\PhiB(X)-\PhiB(Y)}_1}{\norm{X-Y}_2}=U_{\B}^{\R,2}(1).
\]
Therefore we have
\[
\beta_\A^{\R,2}(1)\,\,=\,\,\frac{U_{\A}^{\R,2}(1)}{L_{\A}^{\R,2}(1)}\,\,\ge\,\, \frac{U_{\B}^{\R,2}(1)}{L_{\B}^{\R,2}(1)}\,\,=\,\,\beta_\B^{\R,2}(1).
\]
Thus it is enough to prove the proposition in dimension two.

For $\bm{n}=(n_1,n_2)\in {\mathbb S^1}$, define the matrix $Z_{\bm{n}}$ by
\begin{equation}\label{eq:zerotrace}
Z_{\bm{n}}:=\begin{pmatrix}n_1&n_2\\ n_2&-n_1\end{pmatrix}.
\end{equation}
The eigenvalues of $Z_{\bm{n}}$ are $1$ and $-1$. Set $p_i:=(1/2) \Tr(B_i)$. Then $B_i-p_i I_2$ is of the form $q_iZ_{\bm{n}_i}$ for some $q_i\in [0,p_i]$ and $\bm{n}_i\in \mathbb S^1$. Thus we can parameterize each $B_i\in \mathbb{S}_2^+(\R)$ as
\[
B_i=p_i I_2+q_iZ_{\bm{n}_i}.
\]
Given $X,Y\in \mathcal M_1(\R)$ with $\norm{X-Y}_2=1$, assume that one eigenvalue of $X-Y$ is $1$, otherwise interchange $X$ and $Y$. Similarly, such $X-Y$ can be parameterized as
\[
X-Y=tI_2+(1-t)Z_{\bm{u}},
\]
where $t\in[0,1/2]$, $\bm{u}\in \mathbb S^1$, and $Z_{\bm{u}}$ is the matrix generated by $\bm{u}$ according to \eqref{eq:zerotrace}. Moreover, for every $t\in[0,1/2]$ and $\bm{u}\in \mathbb S^1$, $tI_2+(1-t)Z_{\bm{u}}$ is of the form $X-Y$ for some $X,Y\in \mathcal M_1(\R)$. Since $\ip{Z_{\bm{n}_i}}{Z_{\bm{u}}}=2\ip{\bm{n}_i}{\bm{u}}$, the parametrizations of $B_i$ and $X-Y$ give
\begin{equation}\label{eq:rankratio}
\frac{\norm{\Phi_{\B}(X)-\Phi_{\B}(Y)}_1}{\norm{X-Y}_2}
=
\sum_{i=1}^m \abs{2p_i t+2q_i(1-t)\ip{\bm{n}_i}{\bm{u}}}.
\end{equation}
Thus we have
\[
U_{\B}^{\R,2}(1)= \sup_{t,\bm{u}} \sum_{i=1}^m \abs{2p_i t+2q_i(1-t)\ip{\bm{n}_i}{\bm{u}}}, \qquad
L_{\B}^{\R,2}(1)= \inf_{t,\bm{u}} \sum_{i=1}^m \abs{2p_i t+2q_i(1-t)\ip{\bm{n}_i}{\bm{u}}}.
\]

First we test two convenient values of $t$ for $U_{\B}^{\R,2}(1)$. Taking $t=1/2$ and averaging \eqref{eq:rankratio} over $\bm{u}\in \mathbb S^1$, we obtain
\begin{equation}\label{eq:traceaverage}
U_{\B}^{\R,2}(1)\ge \frac{1}{2\pi}\int_{\mathbb S^1}\sum_{i=1}^m\abs{2p_i t+2q_i(1-t)\ip{\bm{n}_i}{\bm{u}}} d\bm{u}\ge \sum_{i=1}^m p_i+\frac{1}{2\pi}\sum_{i=1}^m q_i\int_{\mathbb S^1}\ip{\bm{n}_i}{\bm{u}}d\bm{u}=\sum_{i=1}^m p_i.
\end{equation}
Then taking $t=0$ and averaging again, we get
\begin{equation}\label{eq:circleaverage}
U_{\B}^{\R,2}(1)\ge\frac{1}{2\pi}\int_{\mathbb S^1}\sum_{i=1}^m\abs{2p_i t+2q_i(1-t)\ip{\bm{n}_i}{\bm{u}}} d\bm{u}= \frac{1}{2\pi}\sum_{i=1}^m 2q_i \int_0^{2\pi} \abs{\cos\theta} d\theta = \frac{4}{\pi} \sum_{i=1}^m q_i.
\end{equation}
Combining \eqref{eq:traceaverage} and \eqref{eq:circleaverage} gives
\begin{equation}\label{eq:rankupper}
U_{\B}^{\R,2}(1)\ge \max\left\{\sum_{i=1}^m p_i,\frac{4}{\pi}\sum_{i=1}^m q_i\right\}.
\end{equation}
Fixing $t\in[0,1/2]$ and averaging over the circle, we have
\begin{equation}\label{eq:rankaverage}
\begin{split}
L_{\B}^{\R,2}(1) &\le \frac{1}{2\pi}\int_0^{2\pi}\sum_{i=1}^m \left|2p_i t+2q_i(1-t)\cos\theta\right| \,d\theta \\
&\le \frac{1}{2\pi}\int_0^{2\pi}\sum_{i=1}^m (2(p_i-q_i)t+2q_i\abs{t+(1-t)\cos\theta}) d\theta \\
&\le 2\left(\sum_{i=1}^m p_i-\sum_{i=1}^m q_i\right)t+\sum_{i=1}^m q_i \frac{4}{\pi}\left(\sqrt{1-2t}+t\arcsin\frac{t}{1-t}\right).
\end{split}
\end{equation}
It is convenient to denote
\[
\begin{split}
g(t)&:=2\left(\sum_{i=1}^m p_i-\sum_{i=1}^m q_i\right)t+\sum_{i=1}^m q_i \frac{4}{\pi} \left(\sqrt{1-2t}+t\arcsin\frac{t}{1-t} \right), \\
m(t)&:=\left(2-\frac{\pi}{2}\right)t+\left(\sqrt{1-2t}+t\arcsin\frac{t}{1-t}\right).
\end{split}
\]
We claim that for $t\in [0,1/2]$,
\begin{equation}\label{eq:anglebound}
t\le \frac{2}{\pi}\left(\sqrt{1-2t}+t\arcsin\frac{t}{1-t}\right).
\end{equation}
Indeed, for $0\le t\le 1/2$, let $\theta:=\arcsin(t/(1-t))$. Then, $0\le \theta\le \pi/2$, $t=\sin\theta/(1+\sin\theta)$, and $\sqrt{1-2t}=\cos\theta/(1+\sin\theta)$.
Multiplying both sides of \eqref{eq:anglebound} by $\pi(1+\sin\theta)/2$ transforms the inequality into $\pi\sin\theta /2\le \cos\theta+\theta\sin\theta$. Define $ h(\theta):=\cos\theta+(\theta-\pi/2)\sin\theta$. $h'(\theta)=(\theta-\pi/2)\cos\theta\le 0$ for $0\le \theta\le \pi/2$, and $h(\pi/2)=0$. Hence $h(\theta)\ge0$ for $0\le \theta\le \pi/2$, which proves the claim.

If $\sum_{i=1}^m q_i\le (\pi/4)\sum_{i=1}^m p_i$, then \eqref{eq:anglebound} implies
\[
g(t)-\sum_{i=1}^m p_i m(t) = \left(\frac{\pi}{2}\sum_{i=1}^m p_i-2\sum_{i=1}^m q_i\right) \left(t-\frac{2}{\pi} \left(\sqrt{1-2t}+t\arcsin\frac{t}{1-t}\right)\right)\le 0.
\]
If $\sum_{i=1}^m q_i\ge (\pi/4)\sum_{i=1}^m p_i$, then
\[
g(t) - \frac{4}{\pi}\sum_{i=1}^m q_i m(t)=\left(2\sum_{i=1}^m p_i-\frac{8}{\pi}\sum_{i=1}^m q_i\right)t\le 0.
\]
Thus in either case, we have
\begin{equation}\label{eq:weightedbound}
g(t)\le \max\left\{\sum_{i=1}^m p_i,\frac{4}{\pi}\sum_{i=1}^m q_i\right\}m(t).
\end{equation}
Combining \eqref{eq:rankupper}, \eqref{eq:rankaverage}, and \eqref{eq:weightedbound}, we obtain
\begin{equation}\label{eq:ratiobound}
\beta_{\B}^{\R,2}(1)=\frac{U_{\B}^{\R,2}(1)}{L_{\B}^{\R,2}(1)}\ge \frac{1}{m(t)}.
\end{equation}

We now minimize $m(t)$ over $t\in[0,1/2]$. Set $\theta=\arcsin(t/(1-t))$. As $0\le t\le 1/2$, $0\le\theta\le\pi/2$ and $m(t)$ becomes
\[
m(t)=\frac{\cos\theta+\left(\theta+2-\pi/2\right)\sin\theta}{1+\sin\theta}=:M(\theta).
\]
Differentiating gives
\[
M'(\theta)=\frac{\cos\theta \left(\theta+2-\pi/2-\cos\theta\right)}{(1+\sin\theta)^2}.
\]
Since $\theta+2-\pi/2-\cos\theta$ is strictly increasing on $[0,\pi/2]$, negative at $0$, and positive at $\pi/2$, it has a unique zero, denoted by $\vartheta_0$. Thus $M(\theta)$ decreases on $[0,\vartheta_0]$ and increases on $[\vartheta_0,\pi/2]$. Hence the minimum is attained at $\vartheta_0$. Using $
\vartheta_0+2-\pi/2=\cos\vartheta_0$, we obtain
\begin{equation}\label{eq:angleminimum}
\min_{0\le t\le 1/2} m(t)=M(\vartheta_0)=\frac{\cos\vartheta_0+\cos\vartheta_0\sin\vartheta_0}{1+\sin\vartheta_0}
=
\cos\vartheta_0.
\end{equation}
Substituting \eqref{eq:angleminimum} into \eqref{eq:ratiobound}, we conclude that
\begin{equation}\label{eq:sharpratio}
\beta_{\B}^{\R,2}(1)\ge \sec\vartheta_0.
\end{equation}

We now construct a family whose condition number approaches the lower bound. Define $A^{(N)}=\{A_k^{(N)}\}_{k=0}^N$ by
\[
A_0^{(N)}:=\left(1-\frac{\pi}{4}\right)I_2,
\qquad
A_k^{(N)}:=\frac{\pi}{2N}\bm{u}_k \bm{u}_k^\top \quad\text{for } 1\le k\le N
\]
with $\bm{u}_k:= (\cos((k-1)\pi/N), \sin((k-1)\pi/N))^T$. For $1\le k\le N$, set $\bm{n}_k:=(\cos(2(k-1)\pi/N), \sin(2(k-1)\pi/N))$. Then we have
\[
A_k^{(N)}=\frac{\pi}{4N}(I_2+Z_{\bm{n}_k}).
\]
Suppose that $X-Y=tI_2+(1-t)Z_{\bm{u}_\theta}$, where $\bm{u}_\theta:=(\cos\theta,\sin\theta)$. Substituting into \eqref{eq:rankratio} yields
\begin{equation}\label{eq:circlesum}
\frac{\norm{\Phi_{A^{(N)}}(X)-\Phi_{A^{(N)}}(Y)}_1}{\norm{X-Y}_2}
=
\left(2-\frac{\pi}{2}\right)t+\frac{\pi}{2N}\sum_{k=0}^{N-1}\abs{t+(1-t)\cos\left(\theta-\frac{2\pi k}{N}\right)}.
\end{equation}
Since
\[
m(t)=\left(2-\frac{\pi}{2}\right)t+\frac14 \int_0^{2\pi}\abs{t+(1-t)\cos\theta} d\theta,
\]
the right-hand side of \eqref{eq:circlesum} is a Riemann sum approximating $m(t)$. Convergence of the Riemann sum gives
\begin{equation}\label{eq:uniformlimit}
\lim_{N\to \infty} \sup_{\substack{0\le t\le \frac12\\ \theta\in\R}}\abs{
\left(2-\frac{\pi}{2}\right)t
+\frac{\pi}{2N}\sum_{k=0}^{N-1}\abs{t+(1-t)\cos\left(\theta-\frac{2\pi k}{N}\right)}-m(t)}
=0.
\end{equation}
In particular,
\[
\lim_{N\to \infty} U_{A^{(N)}}^{\R,2}(1)=\lim_{N\to \infty}\sup_{\theta,t}\ \left(2-\frac{\pi}{2}\right)t
+\frac{\pi}{2N}\sum_{k=0}^{N-1}\abs{t+(1-t)\cos\left(\theta-\frac{2\pi k}{N}\right)}
=
\max_{0\le t\le \frac12}m(t)=1,
\]
and
\[
\lim_{N\to \infty} L_{A^{(N)}}^{\R,2}(1)=\lim_{N\to \infty}\inf_{\theta,t}\  \left(2-\frac{\pi}{2}\right)t
+\frac{\pi}{2N}\sum_{k=0}^{N-1}\abs{t+(1-t)\cos\left(\theta-\frac{2\pi k}{N}\right)}
=
\min_{0\le t\le \frac12}m(t)=\cos\vartheta_0.
\]
Hence
\[
\lim_{N\to \infty} \beta_{A^{(N)}}^{\R,2}(1)
=
\frac{1}{\cos\vartheta_0}=\sec\vartheta_0.
\]
This shows that the bound is asymptotically sharp.
\end{proof}

\subsection{Proof of Lemma~\ref{lem:meanlower}}
\begin{lemma}\label{lem:smallball}\textnormal{\cite[Lemma~8]{KrahmerStoger2020}}
Let $\bm{a}\in\HH^n$ be a standard Gaussian vector and $H\in \Sn(\mathbb H)$. Then we have
\[
\Pbb\left\{\abs{\bm{a}^*H\bm{a}}^2\ge \frac{\E \abs{\bm{a}^*H\bm{a}}^2}{2}\right\}\gtrsim \frac{\left(\E \abs{\bm{a}^*H\bm{a}}^2\right)^2}{\norm{H}_F^4+\Tr(H)^4}.
\]
\end{lemma}

\begin{proof}[Proof of Lemma~\ref{lem:meanlower}]
Denote the eigenvalues of $H$ by $\lambda_1,\dots,\lambda_n$. Then $\bm{a}^* H\bm{a} \stackrel{d}{=}  \sum_{j=1}^n \lambda_j \abs{a_j}^2$. Since $\E \abs{a_j}^2=1$, we have
\[
\E \abs{\bm{a}^* H\bm{a}} \ge \abs{\E\bm{a}^* H\bm{a}}=\Big|\sum_{j=1}^n \lambda_j\Big|=\abs{\Tr{H}}.
\]
Also, $\E \abs{a_i}^4=3$ if $\HH=\R$, and $\E \abs{a_i}^4=2$ if $\HH=\C$. Thus $\E(\bm{a}^* H\bm{a})^2= 2\norm{H}_F^2+\Tr(H)^2$ if $\HH=\R$, and $\E(\bm{a}^* H\bm{a})^2= \norm{H}_F^2+\Tr(H)^2$ if $\HH=\C$. Substituting into Lemma~\ref{lem:smallball} yields
\[
\Pbb\left\{\abs{\bm{a}^*H\bm{a}}\ge \frac{\norm{H}_F}{\sqrt2}\right\}\gtrsim 1,
\]
which proves $\E \abs{\bm{a}^*H\bm{a}}\gtrsim \norm{H}_F$.
\end{proof}

\subsection{Proof of Lemma~\ref{lem:matrixmoment}}
\begin{lemma}\label{lem:ranktail}\textnormal{\cite[Lemma~7]{GaoMaRenZhou2015}}
Let $\bm{a}_1,\ldots,\bm{a}_m$ be independent $N(0,I_n)$ random vectors in
$\R^n$. Then there exist universal constants $C,c>0$ such
that, for every integer $1\le r\le n$ and every $t>0$,
\[
\Pbb\bigg\{\sup_{\substack{\rank(K)\le r\\ \norm{K}_F\le 1}}
\bigg|\Big\langle\frac1m\sum_{i=1}^m \bm{a}_i\bm{a}_i^\top-I_n,K\Big\rangle\bigg|>t\bigg\}
\le
e^{Cnr-cm\cdot\min(t^2,t)}.
\]
\end{lemma}

\begin{proof}[Proof of Lemma~\ref{lem:matrixmoment}]

Von Neumann's trace inequality yields
\begin{equation}\label{eq:rankduality}
\norm{\sum_{i=1}^m(\bm{a}_i \bm{a}_i^*-I_n)}_{(r,2)}=\sup_{\substack{\operatorname{rank}(K)\le r\\ \norm{K}_F\le1}}\abs{\Big\langle\sum_{i=1}^m(\bm{a}_i \bm{a}_i^*-I_n),K\Big\rangle}.
\end{equation}
We first prove the real case. Using Lemma~\ref{lem:ranktail} with $t=C'(\sqrt{(nr+u)/m}+(nr+u)/m)$ for sufficiently large $C'$ yields
\[
\Pbb\left\{\Big\|\sum_{i=1}^m(\bm{a}_i \bm{a}_i^*-I_n)\Big\|_{(r,2)}>C\left(\sqrt{m(nr+u)}+nr+u\right)\right\}\le e^{-cu}.
\]
Since $\sqrt{m(nr+u)}\le \sqrt{mnr}+\sqrt{mu}$, Lemma~\ref{lem:tailmoment} gives
\begin{equation}\label{eq:rankmoment}
\left(\E\Big\|\sum_{i=1}^m(\bm{a}_i \bm{a}_i^*-I_n)\Big\|_{(r,2)}^p\right)^{1/p} \le C\left(\sqrt{mnr}+nr+\sqrt{mp}+p\right).
\end{equation}
This proves \eqref{eq:matrixmoment} for $\mathbb H=\R$.

We now prove the complex case by realification.  Write $\bm{a}_i=\bm{x}_i+\mathrm{i}\bm{y}_i$ where $\bm{x}_i,\bm{y}_i\in\R^n$, and set
\[
\bm{g}_i:=\sqrt2\binom{\bm{x}_i}{\bm{y}_i}\in\R^{2n}.
\]
Since $\bm{a}_i$ is a standard complex Gaussian vector, each $\bm{g}_i$ is a standard real Gaussian vector in $\R^{2n}$.
For a Hermitian matrix $K=A+\mathrm{i}B$, we define
\[
\mathcal R(K):=
\begin{pmatrix}
A & -B\\
B & A
\end{pmatrix}.
\]
Then $\mathcal R(K)$ is real symmetric and satisfies
\(
\ip{\bm{a}_i\bm{a}_i^*-I_n}{K}=\frac12 \ip{\bm{g}_i\bm{g}_i^{\top}-I_{2n}}{\mathcal R(K)}.
\)
Moreover, $\operatorname{rank}(\mathcal R(K))=2\operatorname{rank}(K)$ and $\norm{\mathcal R(K)}_F=\sqrt2 \norm{K}_F$ .
Using \eqref{eq:rankduality} again,
\[
\begin{aligned}
\norm{\sum_{i=1}^m(\bm{a}_i \bm{a}_i^*-I_n)}_{(r,2)} &=
\sup_{\substack{\operatorname{rank}(K)\le r\\ \norm{K}_F\le1}} \abs{\Big\langle\sum_{i=1}^m\left(\bm{a}_i \bm{a}_i^*-I_n\right),K\Big\rangle}
=\frac12 \sup_{\substack{\operatorname{rank}(K)\le r\\ \norm{K}_F\le1}}
\abs{\Big\langle\sum_{i=1}^m\left(\bm{g}_i \bm{g}_i^\top-I_{2n}\right),\mathcal R(K)\Big\rangle} \\
&\le \frac1{\sqrt2} \sup_{\substack{\operatorname{rank}(L)\le 2r\\ \norm{L}_F\le1}} \abs{\Big\langle\sum_{i=1}^m\left(\bm{g}_i \bm{g}_i^\top-I_{2n}\right),L\Big\rangle}
=\frac1{\sqrt2} \norm{\sum_{i=1}^m\left(\bm{g}_i \bm{g}_i^\top-I_{2n}\right)}_{(2r,2)}.
\end{aligned}
\]
Applying the real case result in dimension $2n$ with rank parameter $2r$ gives that \eqref{eq:rankmoment} also holds in the complex case.
\end{proof}

\subsection{Proof of Lemma~\ref{lem:gammagap} }
\begin{proof}[Proof of Lemma~\ref{lem:gammagap}]
Define $T:=G_1+G_2$, and $U:=G_1/(G_1+G_2)$.
By Lemma~\ref{lem:gammaratio}, $T$ and $U$ are independent, with $T\sim \Gamma(2s,\theta)$ and $U\sim \mathrm{Beta}(s,s)$. Since $|G_1-G_2|=T|2U-1|$, we obtain $\E\abs{G_1-G_2}=\E[T]\cdot \E\abs{2U-1}$. Note that $\E[T]=2s\theta$, and a direct computation gives
\[
\E\abs{G_1-G_2}=\frac{4s\theta\Gamma(2s)}{\Gamma(s)^2}\int_{1/2}^{1}(2u-1)u^{s-1}(1-u)^{s-1}\,du
=
\frac{2\theta 2^{1-2s}\Gamma(2s)}{\Gamma(s)^2} =\frac{2\theta\Gamma\left((2s+1)/2\right)}{\sqrt{\pi}\Gamma(s)}.
\]

In the special case where $G_1\sim\chi^2(2)$ and $G_2\sim\chi^2(1)$, their respective densities are $f_1(x)=e^{-x/2}/2$ and $f_2(x)=e^{-x/2}/\sqrt{2\pi x}$ for $x\ge 0$.
For fixed $t\ge 0$, a direct computation gives
\[
\E\abs{G_1-t}=\int_0^t (t-x)\frac12 e^{-x/2} dx+\int_t^\infty (x-t)\frac12 e^{-x/2} dx=t-2+4e^{-t/2}.
\]
Therefore, we have
\[
\E\abs{G_1-G_2}=\E\left(G_2-2+4e^{-G_2/2}\right)=-1+4\int_0^\infty \frac{1}{\sqrt{2\pi x}}e^{-x} dx=2\sqrt2-1.
\]
\end{proof}

\bibliographystyle{abbrv}
\bibliography{reference}

\end{document}